\documentclass[12pt,draftclsnofoot,letterpaper,onecolumn]{IEEEtran}
\usepackage[compress]{cite}
\usepackage[utf8]{inputenc}
\usepackage{amsmath}
\usepackage{amssymb}
\usepackage{mathrsfs}
\usepackage{enumerate}
\usepackage{adjustbox}
\usepackage{xcolor}
\usepackage{graphicx}
\usepackage{algorithm}
\usepackage{algpseudocode}
\usepackage{xpatch}
\usepackage{mathtools}
\usepackage{nicefrac}
\usepackage{float}
\floatstyle{plaintop}
\restylefloat{table}
\usepackage{cleveref}

\algdef{SE}[DOWHILE]{Do}{doWhile}{\algorithmicdo}[1]{\algorithmicwhile\ #1}%
\algrenewcommand\alglinenumber[1]{{\sffamily\footnotesize#1}}
\makeatletter
\xpatchcmd{\algorithmic}{\itemsep\z@}{\itemsep=.25ex plus2pt}{}{}
\makeatother

\newtheorem{Lemma}{Lemma}
\newtheorem{Remark}{Remark}

\DeclareMathOperator*{\argmax}{arg\,max}

\usepackage{color}
\usepackage{multirow} 
\usepackage{soul}
\usepackage{environ}         % provides \BODY
\usepackage{etoolbox}        % provides \ifdimcomp
\usepackage[acronym]{glossaries}
\newacronym{1g}{1G}{first-generation}
\newacronym{4g}{4G}{fourth-generation}
\newacronym{5g}{5G}{fifth-generation}
\newacronym{pca}{PCA}{principal component analysis}
\newacronym{iot}{IoT}{Internet of Things}
\newacronym{mimo}{MIMO}{multiple-input multiple-output}
\newacronym{ris}{RIS}{reconfigurable intelligent surface}
\newacronym{siso}{SISO}{single-input-single-output}
\newacronym{dnn}{DNN}{deep neural network} 
\newacronym{lstm}{LSTM}{long short-term memory}
\newacronym{lmmse}{LMMSE}{linear minimum mean-squared error}
\newacronym{gpu}{GPU}{graphic processing unit}
\newacronym{cb}{CB}{conjugate beamforming}
\newacronym{mamimo}{mMIMO}{massive multiple-input multiple-output}
\newacronym{cf-mmimo}{CF-mMIMO}{cell-free mMIMO}
\newacronym{sumimo}{SU-MIMO}{single user MIMO}
\newacronym{mumimo}{MU-MIMO}{multi user MIMO}
\newacronym{embms}{eMBMS}{evolved Multimedia Broadcast and Multicast Service}
\newacronym{sca}{SCA}{successive convex approximation}
\newacronym{sinr}{SINR}{signal-to-interference-plus-noise ratio}
\newacronym{ula}{ULA}{uniform linear array}
\newacronym{mcs}{MCS}{modulation and coding scheme}
\newacronym{mrt}{MRT}{maximum ratio transmission}
\newacronym{zf}{ZF}{zero-forcing}
\newacronym{mr}{MR}{maximum ratio}
\newacronym{se}{SE}{spectral efficiency}
\newacronym{sse}{SumSE}{sum spectral efficiency}
\newacronym{mise}{MinSE}{minimum spectral efficiency}
\newacronym{asd}{ASD}{angular standard deviation}
\newacronym{adr}{ADR}{aggregated data rate}
\newacronym{embb}{eMBB}{enhanced mobile broadband}
\newacronym{mmtc}{mMTC}{massive machine type communications}
\newacronym{urllc}{URLLC}{ultra reliable low latency communications}
\newacronym{csi}{CSI}{channel state information}
\newacronym{pmi}{PMI}{precoding matrix indicator}
\newacronym{ri}{RI}{rank indicator}
\newacronym{csi-rs}{CSI-RS}{CSI-reference signal}
\newacronym{cri}{CRI}{CSI-RS resource indicator}
\newacronym{bs}{BS}{base station}
\newacronym{re}{RE}{resource element}
\newacronym{mmwave}{mmWave}{millimeter-wave}
\newacronym{umwave}{$\mu$mWaves}{micrometer waves}
\newacronym{rnn}{RNN}{recurrent neural network}
\newacronym{cnn}{CNN}{convolutional neural network}
\newacronym{ngmn}{NGMN}{next-generation mobile network}
\newacronym{lte}{LTE}{Long Term Evolution}
\newacronym{lte-a}{LTE-A}{Long Term Evolution Advanced}
\newacronym{5gnr}{5G NR}{5G New Radio}
\newacronym{mm}{MM}{mixed mode}
\newacronym{cdf}{CDF}{cumulative distribution function}
\newacronym{phy}{PHY}{physical}
\newacronym{mac}{MAC}{medium access control}
\newacronym{3gpp}{3GPP}{3rd Generation Partnership Project}
\newacronym{fdd}{FDD}{frequency division duplexing}
\newacronym{tdd}{TDD}{time division duplexing}
\newacronym{ofdm}{OFDM}{orthogonal frequency division multiplexing}
\newacronym{ss}{SS}{synchronization signal} 
\newacronym{pss}{PSS}{primary synchronization signal} 
\newacronym{sss}{SSS}{secondary synchronization signal} 
\newacronym{pbch}{PBCH}{physical broadcast channel} 
\newacronym{dmrs}{DMRS}{demodulation reference signal} 
\newacronym{gnb}{gNB}{next generation nodeB} 
\newacronym{rsrp}{RSRP}{reference signal received power} 
\newacronym{rrm}{RRM}{radio resource management} 
\newacronym{srs}{SRS}{sounding reference signal} 
\newacronym{ran}{RAN}{radio access network} 
\newacronym{nn}{NN}{neural network} 
\newacronym{ue}{UE}{user equipment} 
\newacronym{awgn}{AWGN}{additive white Gaussian noise} 
\newacronym{epa}{EPA}{Extended Pedestrian A model}
\newacronym{eva}{EVA}{Extended Vehicular A model}
\newacronym{etu}{ETU}{Extended Typical Urban model}
\newacronym{tdl}{TDL}{tapped delay line}
\newacronym{tl}{TL}{transfer learning}
\newacronym{cdl}{CDL}{clustered delay line}
\newacronym{uma}{UMa}{urban macro-cell}
\newacronym{isd}{ISD}{inter-site distance}
\newacronym{nlos}{NLOS}{non-line of sight}
\newacronym{los}{LOS}{line of sight}
\newacronym{o2o}{O2O}{outdoor-to-outdoor}
\newacronym{o2i}{O2I}{outdoor-to-indoor}
\newacronym{ul}{UL}{uplink}
\newacronym{dl}{DL}{downlink}
\newacronym{ls}{LS}{least squares}
\newacronym{mmse}{MMSE}{minimum mean square error}
\newacronym{snr}{SNR}{signal-to-noise ratio}
\newacronym{mse}{MSE}{mean square error}
\newacronym{nr}{NR}{New Radio}
\newacronym{prb}{PRB}{physical resource block}
\newacronym{scs}{SCS}{subcarrier spacing}
\newacronym{bler}{BLER}{block error rate}
\newacronym{smmmra}{SMMMRA}{subgroup multicast \gls{mamimo} resource allocation}
\newacronym{mmf}{MMF}{max-min fairness}
\newacronym{smmu}{SMMU}{subgroups of multicast \gls{mamimo} users}
\newacronym{gsmma}{GSMMA}{greedy subgroup multicast \gls{mamimo} algorithm}
\newacronym{sdr}{SDR}{semidefinite relaxation}
\newacronym{wmmse}{WMMSE}{weighted minimum mean square error}

\usepackage{flushend}
\newcommand{\herm}{^\mathsf{H}}
\newcommand{\trans}{^\mathsf{T}}
\usepackage{tikz}
\usepackage{changepage}
\usetikzlibrary{fadings}
\tikzfading[name=fade out, 
    inner color=transparent!0,
    outer color=transparent!100]
\usetikzlibrary{patterns}
\usetikzlibrary{shadows.blur}
\usetikzlibrary{shapes}
\usetikzlibrary{arrows}

\ifCLASSOPTIONcompsoc
    \usepackage[caption=false, font=normalsize, labelfont=sf, textfont=sf, subrefformat=parens, labelformat=parens]{subfig}
\else
\usepackage[caption=false, font=footnotesize, subrefformat=parens, labelformat=parens]{subfig}
\fi

\begin{document}
\title{Deep Learning-Assisted Multicast Subgrouping in Massive MIMO}

\author{Alejandro de la Fuente, Leopoldo Carro-Calvo, and Giovanni Interdonato~\IEEEmembership{Member,~IEEE}

\thanks{A. de la Fuente and L. Carro-Calvo are with the Department of Signal Theory and Communications; L. Carro-Calvo is also with Institute of Data, Complex Networks and Cybersecurity Sciences (DCNC Sciences), Universidad Rey Juan Carlos, Camino del Molino, 28943, Fuenlabrada (Madrid), Spain (email: alejandro.fuente, leopoldo.carro@urjc.es).

G. Interdonato is with Ericsson Research, Ericsson AB, 164 83, Stockholm, Sweden. G. Interdonato was with the Department of Electrical and Information Engineering (DIEI), University of Cassino and Southern Lazio, 03043 Cassino, Italy.}

\thanks{This work was supported in part by SOFIA-WIND funded by MICIU/AEI/10.13039/501100011033 and ERDF, EU under Grant PID2023-147305OB-C33; in part by brAIn5G funded by MICIU/AEI/10.13039/501100011033 and FEDER under Grant PID2024-161515OA-I00; in part by POLIGRAPH funded by MCIU/AEI/10.13039/501100011033 and by the ``European Union NextGenerationEU/PRTR'' under the Grant PID2022-136887NB-I00; in part by  AI-XCAST6G funded by the IMPULSO program for research at the Rey Juan Carlos University under Grant ref F1263.}}

%\markboth{IEEE copyright. This is an author-created postprint version. The final publication is available at https://ieeexplore.ieee.org/xxxx . A. de la Fuente, L. Carro-Calvo, and G. Interdonato “Deep Learning-Assisted Multicast Subgrouping in Massive MIMO”, IEEE Open Journal of the Communications Society 7, xxx-xxx}

\maketitle

\begingroup
\renewcommand\thefootnote{}\footnote{%
\textit{This is an author-created postprint version. 
The final publication is available at https://ieeexplore.ieee.org/xxxx. 
\copyright\ 2026 IEEE. Personal use is permitted. 
Citation: A. de la Fuente, L. Carro-Calvo, and G. Interdonato, 
“Deep Learning-Assisted Multicast Subgrouping in Massive MIMO,” 
IEEE Open Journal of the Communications Society, vol. 7, pp. xxx--xxx.}
}
\addtocounter{footnote}{-1}
\endgroup

\begin{abstract}
Efficient content delivery in massive multiple-input multiple-output (mMIMO) multicasting is fundamentally limited by pilot overhead and the need to serve heterogeneous users with a common transmission rate. Conventional approaches either suffer from pilot contamination or are constrained by the worst-user effect, motivating the need for adaptive subgrouping strategies.
In this paper, we propose a deep learning-assisted multicast subgrouping framework that infers the number of multicast subgroups directly from users’ spatial channel statistics. A snapshot-specific principal component analysis (PCA) is applied to user covariance matrices to obtain a compact representation, which is processed by a sequential long short-term memory (LSTM) encoder capable of handling variable-size user sets. The model predicts the number of subgroups and groups users based on their statistical similarity.
To further improve system performance, we introduce a transfer learning (TL) extension where a pre-trained LSTM encoder is reused and a lightweight dense head is fine-tuned to estimate the sum spectral efficiency (SE) as a function of the subgroup configuration. This enables selecting near-optimal subgrouping solutions without exhaustive search.
Simulation results demonstrate that the proposed approach consistently outperforms benchmark methods, including unicast transmission, conventional multicast, random subgrouping, and density-based clustering. The TL-enhanced model achieves up to $85\%$ of the maximum achievable spectral efficiency while maintaining robust performance across diverse spatial user distributions and under imperfect covariance information.
\end{abstract}


\maketitle

\glsresetall

\section{INTRODUCTION}

{\color{black}\IEEEPARstart{T}{he} mobile network data traffic has increased by more than an order of magnitude over the past decade and continues to grow steadily, reaching on the order of 150--200~EB per month by 2024--2025~\cite{2025Ericsson}. This growth is expected to persist over the remainder of the decade, driven by the increasing number of smartphone subscriptions and the rise of bandwidth-intensive services such as mobile video, which already accounts for more than 70\% of total traffic. Commercial \gls{5g} deployments are accelerating this trend, and by the end of 2029, \Gls{5g} is projected to surpass 5.6~billion subscriptions worldwide, representing more than 60\% of total mobile lines~\cite{2025Ericsson}.}

Within this rapidly expanding traffic volume, a substantial share corresponds to content that can be consumed by multiple users concurrently; hence, broadcast and multicast mechanisms are natural enablers to exploit such group demand \cite{2016delaFuente,2017Araniti}. In particular, multicasting is pivotal for group-oriented services such as live video streaming, virtual/augmented reality, software updates, and \gls{iot} use cases.

\Gls{mamimo} is a key technology for beyond-\gls{5g} systems, with emerging paradigms such as cell-free architectures \cite{2024Ngo,2026Buzzi}, holographic MIMO \cite{2020Huang}, and reconfigurable intelligent surfaces \cite{2018Hu,2020DiRenzo}, offering substantial gains in spectral and energy efficiency due to favorable propagation and channel hardening \cite{2016Bjornson,2019Bjornson,2020Zhang}. As a result, combining multicasting with \gls{mamimo} enables efficient large-scale content delivery \cite{2019Dong}. However, efficiently serving multicast users remains challenging, particularly regarding how to group users according to their channel characteristics.

Multicast subgrouping addresses this challenge by partitioning users requesting the same content into subsets with similar propagation properties \cite{2022delaFuente,2018Riera,2024delaFuente}. This enables the design of multiple precoders for the same multicast stream, improving spectral efficiency while preserving multicast semantics.

{\color{black}In this work, multicast subgrouping is formulated as a two-step process: first, the number of multicast subgroups is inferred from spatial covariance features, and second, users are assigned to subgroups based on covariance similarity. In particular, estimating the number of multicast subgroups constitutes the central learning task, as it directly determines the structure of the subgrouping solution.}

{\color{black}Importantly, subgrouping does not alter the multicast service itself, as all users receive the same data stream. Instead, it enables the design of multiple precoders tailored to statistically homogeneous subsets of users, thereby mitigating the worst-user limitation inherent to conventional multicast transmission. This approach leverages the spatial correlation structure of massive MIMO channels to improve physical-layer efficiency without increasing the transmitted content~\cite{2022delaFuente,2024delaFuente}.}

{\color{black}However, determining the appropriate number of multicast subgroups in a scalable and adaptive manner remains a challenging and largely unresolved problem in scalable massive MIMO multicast systems. This problem becomes especially critical in large-scale multicast scenarios, where the number of users is high and spatial channel conditions are highly heterogeneous.}

{\color{black}This motivates the need to revisit multicast subgrouping under a data-driven perspective, where the subgroup structure is inferred rather than predetermined.}

\subsection{RELATED LITERATURE}

A broad body of work has investigated multicast transmission and precoding in \gls{mamimo} systems \cite{2013Yang,2017SadeghiTWC,2018SadeghiTWC1,2018SadeghiTWC2,2019Dong}. These works essentially adopt two main strategies: unicast transmission per user, which avoids worst-user effects but incurs high pilot overhead and pilot contamination, and single-group multicast, which reduces training overhead but limits performance due to the weakest user.

To overcome these limitations, multicast subgrouping strategies have been proposed \cite{2022delaFuente,2018Riera,2024delaFuente}. These approaches group users based on spatial covariance similarity, typically using clustering techniques such as K-means \cite{2018Riera}. This enables the design of subgroup-specific precoders while maintaining multicast semantics.

{\color{black}In \cite{2022delaFuente}, multicast subgrouping is performed by clustering users according to the similarity of their spatial covariance matrices, using a K-means algorithm. While this approach effectively captures spatial correlation among users, the number of multicast subgroups is selected heuristically and does not adapt to the underlying user distribution. Consequently, determining the appropriate number of multicast subgroups remains a challenging and largely unresolved problem in scalable massive MIMO multicast systems. This limitation motivates the present work, which aims to infer the number of multicast subgroups in a data-driven and scalable manner, thereby enabling adaptive subgroup configurations without relying on heuristic or fixed designs. It is important to distinguish between the number of spatial clusters used in the data generation process and the number of multicast subgroups that maximizes the sum spectral efficiency. These two quantities represent fundamentally different notions, since the former reflects the underlying spatial distribution of users, whereas the latter is governed by system-level trade-offs such as pilot overhead, inter-group interference, and precoding design.
}

{\color{black}The metric $f_{\mathrm{d}}(\boldsymbol{R}_k, \boldsymbol{R}_{k'})$ takes values in $[0,1]$, where $0$ reflects orthogonality and $1$ corresponds to fully aligned covariance structures, with intermediate values capturing gradual levels of similarity. In \cite{2022delaFuente}, this metric is used with K-means clustering to obtain efficient subgroup partitioning, although without an optimal mechanism to determine the subgroup count.}

Optimization-based multicast beamforming methods, such as \gls{sdr}- and \gls{sca}-based approaches \cite{2023Zhang,2026Zaher}, achieve high performance under full instantaneous \gls{csi} assumptions. However, they require one pilot per user and entail high computational complexity, making them unsuitable for scalable \gls{mamimo} scenarios under limited coherence intervals.

Deep learning techniques have been widely explored in wireless communications, supported by advances in parallel computing and hardware accelerators \cite{2017Chen,2020Chen}, including applications in signal detection, decoding, and resource allocation \cite{2017Samuel,2016Nachmani,2017Challita,2022Melgar}. \Glspl{cnn} and \gls{lstm} networks are particularly effective at capturing structured patterns in wireless channels \cite{9276722,9706946,10073857,2020Marinberg,2021Jiang,2024Carro,2020Lianjun,2022Jin,2020Wang,2021Xiao,2024DiGennaro,2026DiGennaro}.

{\color{black}Deep learning is particularly suitable for problems involving variable-size user sets and complex spatial structures. In this context, recurrent architectures such as \gls{lstm} can capture latent dependencies among users, making them a natural candidate for learning subgroup structures from spatial covariance information \cite{Chukhno2021,2024DiGennaro}. In the proposed framework, the LSTM architecture is specifically used to infer the number of multicast subgroups, leveraging its ability to process variable-length user representations and extract latent spatial patterns.}

{\color{black}Recent studies have also investigated learning-based resource allocation and access mechanisms in large-scale and \gls{cf-mmimo} systems \cite{2025Xu,2026Ying}. While these works address complementary problems, they do not consider multicast subgrouping or the estimation of the number of multicast subgroups.}

{\color{black}While the present work focuses on a single-cell massive MIMO setting, the subgrouping problem is primarily driven by the spatial distribution of users. Therefore, the proposed learning-based framework is expected to generalize at the level of subgroup-structure inference across different network architectures, although performance evaluation remains system-dependent.}

{\color{black}From a system-design perspective, assigning orthogonal uplink pilots to all multicast users is not scalable, since the overhead grows linearly with the number of users and exceeds practical coherence intervals. For this reason, this work focuses on a one-pilot-per-subgroup structure, which enables efficient pilot reuse while maintaining manageable training overhead. This makes scalable subgrouping strategies essential for practical massive MIMO multicast deployments.}

{\color{black}Although all users request the same multicast payload, prior work has shown that forming subgroups based on spatial similarity can significantly improve performance by mitigating worst-user limitations~\cite{2016delaFuente}. This motivates the need for scalable and adaptive subgroup selection mechanisms.}

{\color{black}Unlike prior learning-based approaches that directly map scenario-dependent features to a fixed subgroup count, the proposed framework leverages second-order spatial channel statistics and variable-length user representations to address this challenge, allowing the number of multicast subgroups to be inferred in a physically meaningful and scalable manner, while explicitly capturing the underlying spatial structure of the user distribution.}

{\color{black}From a learning perspective, the multicast subgrouping problem can be interpreted as an inference task over structured spatial channel data, where the goal is to extract the latent grouping structure that determines the optimal number of multicast subgroups.}

\subsection{MAIN CONTRIBUTIONS}

We emphasize that the proposed framework performs \emph{intra-service subgrouping} for a single multicast message in \gls{mamimo} systems.

The main contributions are:

\begin{itemize}
    \item {\color{black}An \gls{lstm}-assisted subgrouping mechanism that infers the number of multicast subgroups from spatial covariance structures.}
    \item {\color{black}A preprocessing pipeline based on per-snapshot \gls{pca} to transform variable-size covariance matrices into fixed-length feature representations.}
    \item {\color{black}A customized LSTM architecture that learns latent spatial correlations and maps them to subgroup configurations.}
    \item {\color{black}A transfer-learning extension that directly predicts normalized sum \gls{se}, avoiding exhaustive subgroup search.}
    \item {\color{black}An intra-service subgrouping strategy that preserves multicast semantics while improving physical-layer efficiency.}
    \item {\color{black}An extensive simulation study validating the proposed approach across diverse scenarios.}
\end{itemize}

\subsection{PAPER OUTLINE AND NOTATIONS}
The structure of this paper is as follows. Section~II introduces the system model and outlines the proposed \gls{mamimo} multicast subgrouping framework operating over spatially correlated Rayleigh fading channels. Section~III describes in detail the DL-assisted strategy used to infer the multicast subgrouping. Section~IV presents numerical evaluations that demonstrate the effectiveness of the proposed approach under various scenarios. Finally, Section~V summarizes the main findings and discusses the relevance of forming multicast subgroups based on large-scale channel characteristics in \gls{mamimo} systems, while also highlighting promising directions for future work.

\textit{Notational remark:} Lowercase and uppercase bold symbols denote vectors and matrices, respectively. Calligraphic uppercase letters represent sets; for a set $\mathcal{A}$, its cardinality is written as $|\mathcal{A}|$. The operators $(\cdot)^{\mathsf{T}}$, $(\cdot)^{\ast}$, and $(\cdot)^{\mathsf{H}}$ stand for the transpose, conjugate, and Hermitian transpose. The field of complex numbers is denoted by $\mathbb{C}$. The expectation operator is written as $\mathbb{E}\{\cdot\}$. For a matrix $\boldsymbol{A}$, $\text{tr}(\boldsymbol{A})$ denotes its trace. A circularly symmetric complex Gaussian distribution with mean vector $\boldsymbol{\mu}$ and covariance matrix $\boldsymbol{\Sigma}$ is written as $\mathcal{CN}(\boldsymbol{\mu},\boldsymbol{\Sigma})$. The identity matrix of size $N$ is denoted by $\mathbf{I}_N$. The Euclidean norm of a vector $\boldsymbol{a}$ is expressed as $\|\boldsymbol{a}\|_2$, and $\langle \boldsymbol{A}, \boldsymbol{B} \rangle$ denotes the inner product between matrices $\boldsymbol{A}$ and $\boldsymbol{B}$.

%%%%%%%%%%%%%%%%%%%%%%%%%%%%%%%%%%%%%%%%%%%%%%%%%%%%%%%%%%%%%%%%%%%%%%%%%%%%%%%%%%%%%%%

\section{SYSTEM MODEL}
\label{sec:system_model}
We consider a single-cell \gls{mamimo} system operating in \gls{tdd}, where a base station equipped with $M$ antennas serves $K$ single-antenna users over shared time-frequency resources.

\textcolor{black}{Since our focus is on multicast transmission, which inherently takes place in the downlink, the uplink data phase is not explicitly modeled. Each \gls{tdd} frame is divided into an uplink training interval and a downlink payload interval of lengths $\tau_{\mathrm{p}}$ and $\tau_{\mathrm{d}}$, respectively, with total duration $\tau_{\mathrm{c}}=\tau_{\mathrm{p}}+\tau_{\mathrm{d}}$ matching the coherence block. The channel is assumed to remain frequency-flat and time-invariant over each coherence block~\cite{2017Bjornsonbook}. This standard modeling assumption does not affect the conclusions of the multicast subgrouping analysis and is adopted throughout the paper.}

\subsection{CHANNEL MODEL}
We adopt a standard block-fading channel model in which, within each
time--frequency coherence block, the channel is frequency-flat and
remains constant, while realizations across different blocks are
statistically independent. The channel vector between the \gls{mamimo}
\gls{bs} and user $k$ in an arbitrary coherence block%
\footnote{For clarity, the coherence-block index is omitted.}%
is denoted by $\boldsymbol{h}_{k}\in\mathbb{C}^{M}$ and is modeled as
\[
\boldsymbol{h}_{k}\sim \mathcal{CN}\!\big(\boldsymbol{0}_{M},\,\boldsymbol{R}_{k}\big),
\]
where $\boldsymbol{R}_{k}\in\mathbb{C}^{M\times M}$ is a positive
semi-definite spatial covariance matrix. The corresponding average
channel gain is
\[
\beta_{k}=\frac{1}{M}\,\text{tr}\!\left(\boldsymbol{R}_{k}\right).
\]

\textcolor{black}{Since the second-order channel statistics evolve over large time scales, the set of covariance matrices $\{\boldsymbol{R}_{k}\}_{k\in\mathcal{K}}$ can be estimated from long-term observations across multiple coherence blocks and is treated as known in this work~\cite{Bjornson2016a}. Practical approaches for estimating spatial correlation matrices in massive MIMO systems can be found in~\cite{Bjornson2016a,CaireC17a,Neumann2018,Sanguinetti2019a,Upadhya2018}.}

\subsection{\textcolor{black}{MULTICAST USER SUBGROUPING (SYSTEM ASSUMPTIONS)}}
\label{sec:subgrouping}

{\color{black}For clarity, we summarize the main assumptions of the multicast subgrouping model in the following.

In contrast to conventional multicast transmission, where all users are served by a single beam, the proposed framework enables multiple parallel multicast transmissions tailored to user subsets.}

We consider a multicast subgrouping framework where a common service is delivered to $G \le K$ disjoint subgroups, all receiving the same data stream. Each subgroup is assigned a dedicated uplink pilot and downlink precoder, enabling efficient exploitation of channel-statistics heterogeneity across users.

{\color{black}Under the considered system model, users within the same subgroup reuse a common uplink pilot sequence, allowing the base station to estimate a composite subgroup channel, which serves as a representative channel for precoder design \cite{2022delaFuente}.}

\textcolor{black}{In the considered \gls{mamimo} multicast framework, users assigned to the same multicast subgroup reuse a common uplink pilot sequence and are served through a shared downlink precoding filter. This one-pilot-per-subgroup structure defines the CSI acquisition model assumed throughout the paper and enables the base station to estimate a composite channel representation for each multicast subgroup. This configuration is adopted as a system-level operating assumption.}

\textcolor{black}{Fig.~\ref{fig:system_model} illustrates a representative example of a multicast service operating in a single-cell massive \gls{mimo} deployment, where the same multicast content is delivered through multiple precoders corresponding to different user subgroups.}

\begin{figure}[t]
\centering
\begin{tikzpicture}[
    user/.style={circle,draw,thick,minimum size=6mm},
    antenna/.style={thick},
    beam/.style={fill opacity=0.25,draw opacity=0},
    >=latex
]

\filldraw[draw=black, thick, fill=gray!15]
  (-0.3,-1.7) rectangle (0.3,1.7);

\foreach \y in {-1.5,-1.25,-1.0,-0.75,-0.5,-0.25,0,
                0.25,0.5,0.75,1.0,1.25,1.5} {
  \draw[thick] (0.3,\y) -- (1.0,\y);
}

\node[rotate=90] at (-0.8,0) {BS};

\fill[black!50,beam]
  (1.0,0.9) -- (6.2,1.7) -- (6.2,0.9) -- cycle;
\draw[->,very thick,black!70!black] (1.0,0.9) -- (6.0,1.3);
\node[black!70!black] at (3.8,1.9) {$\mathbf{w}_1 x_1$};

\node[user,fill=black!25] at (6.0,1.3) {};
\node[user,fill=black!25] at (6.6,1.0) {};
\node[user,fill=black!25] at (5.9,0.9) {};
\node[black!70!black] at (6.4,2) {$g=1$};

\fill[green!50,beam]
  (1,0.0) -- (4.3,0.5) -- (4.3,-0.5) -- cycle;

\draw[->,very thick,green!70!black]
  (1,0.0) -- (3.8,0.0);

\node[green!70!black] at (2.4,0.6) {$\mathbf{w}_2 x_2$};

\node[user,fill=green!25] at (3.8,0.0) {};
\node[user,fill=green!25] at (4.3,-0.3) {};
\node[user,fill=green!25] at (3.7,-0.4) {};

\node[green!70!black] at (4,0.6) {$g=2$};

\fill[red!50,beam]
  (1.0,-0.9) -- (6.4,-0.9) -- (6.4,-1.7) -- cycle;
\draw[->,very thick,red!70!black] (1.0,-0.9) -- (6.1,-1.3);
\node[red!70!black] at (3.8,-1.8) {$\mathbf{w}_3 x_3$};

\node[user,fill=red!25] at (6.1,-1.3) {};
\node[user,fill=red!25] at (6.7,-1.6) {};
\node[user,fill=red!25] at (5.6,-1.6) {};
\node[red!70!black] at (6.3,-0.7) {$g=3$};

\end{tikzpicture}
\caption{Multicast subgrouping in a massive MIMO system. Spatially clustered multicast users are served by a base station equipped with a large antenna array. The same multicast content is transmitted through multiple precoded beams $\mathbf{w}_g x_g$, each tailored to a
user subgroup indexed by $g$ and characterized by similar large-scale channel statistics.}
\label{fig:system_model}
\end{figure}
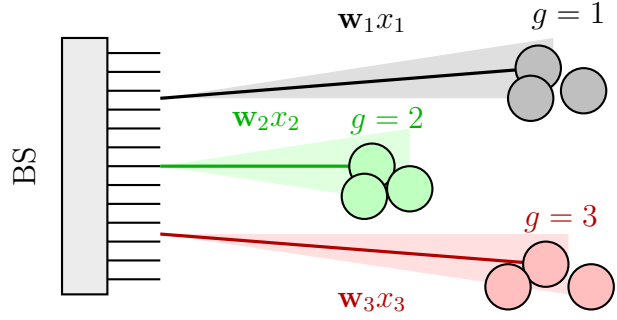

\subsection{UPLINK SUBGROUP CHANNEL ESTIMATION}

Let $\boldsymbol{\psi}_g \in \mathbb{C}^{\tau_{\mathrm{p}}}$ denote the
pilot assigned to subgroup $g$, normalized such that
$\|\boldsymbol{\psi}_g\|^2=1$. Users belonging to the same multicast
subgroup reuse this pilot sequence, while pilots assigned to different
subgroups are mutually orthogonal. Since users sharing the same pilot
(co-pilot users) yield linearly dependent channel estimates, the
\gls{bs} cannot spatially resolve individual users within the same
subgroup~\cite{2016Marzetta,2017Bjornsonbook}.

\textcolor{black}{Without loss of generality, we set $\tau_{\mathrm{p}} = G$, such that $G$ orthogonal uplink pilot sequences are available, thereby avoiding pilot contamination across different
subgroups.}

{\color{black}This structure defines the CSI acquisition model and directly constrains the achievable spatial separability among users, thereby shaping the performance limits of multicast subgrouping.}

Under this configuration, the uplink pilot signal collected at the \gls{bs} over $\tau_{\mathrm{p}}$ channel uses is
\begin{equation}    
    \boldsymbol{Y} 
    = \sqrt{\tau_{\mathrm{p}}}\sum_{g=1}^{G}\sum_{k \in \mathcal{K}_g}
    \sqrt{q_{gk}}\,\boldsymbol{h}_{k}\,\boldsymbol{\psi}\trans_g 
    + \boldsymbol{N},
\end{equation}
where $q_{gk}$ denotes the per-symbol pilot transmit power of user
$k\in\mathcal{K}_g$, and $\boldsymbol{N}\in\mathbb{C}^{M\times
\tau_{\mathrm{p}}}$ is an \gls{awgn} matrix with i.i.d. entries
distributed as $\mathcal{CN}(0,\sigma_{\mathrm{u}}^2)$.

Projecting the received pilot matrix onto $\boldsymbol{\psi}_g^{\ast}$
yields the sufficient statistic for subgroup $g$:
\begin{equation}
\boldsymbol{y}^{g} 
= \sqrt{\tau_{\mathrm{p}}}\sum_{k \in \mathcal{K}_g }
\sqrt{q_{gk}}\,\boldsymbol{h}_{k} 
+ \boldsymbol{n}_g,
\label{eq:y_lg}
\end{equation}
with $\boldsymbol{n}_{g} =
\boldsymbol{N}\boldsymbol{\psi}^{\ast}_g \sim
\mathcal{CN}(0,\sigma_{\mathrm{u}}^2\mathbf{I}_M)$.

We define the subgroup-representative (average) channel as
\begin{equation}
\boldsymbol{h}^g 
= \frac{\sqrt{\tau_{\mathrm{p}}}}{K_g}
\sum_{k \in \mathcal{K}_g}\sqrt{q_{gk}}\,\boldsymbol{h}_{k},
\label{eq:h_gl}
\end{equation}
which is distributed as
$\boldsymbol{h}^g \sim \mathcal{CN}(\boldsymbol{0}_M,\boldsymbol{R}^g)$,
with covariance matrix
\begin{equation}
\boldsymbol{R}^g 
= \frac{\tau_{\mathrm{p}}}{K_g^2}
\sum_{k \in \mathcal{K}_g} q_{gk}\,\boldsymbol{R}_{k}.
\label{eq:R_lg}
\end{equation}

{\color{black}\begin{Remark} The factor $\sqrt{\tau_{\mathrm{p}}}$ in~\eqref{eq:y_lg}
results from projecting the received uplink signal onto the unit-energy
pilot sequence $\boldsymbol{\psi}_g$. Since $\|\boldsymbol{\psi}_g\|^2=1$
and the pilot spans $\tau_{\mathrm{p}}$ samples, this projection yields
an amplitude scaling of $\sqrt{\tau_{\mathrm{p}}}$, consistent with
standard MMSE channel estimation for pilot-based massive MIMO
systems~\cite{2017Bjornsonbook}.
\end{Remark}
}

\begin{Remark}
The ``strong correlation'' among users within the same multicast subgroup
refers to the \emph{similarity of their spatial covariance matrices}, not
to cross-correlation of instantaneous channel realizations. This modeling
assumption is standard in massive MIMO systems
\cite{2017Bjornsonbook,2016Marzetta,2018Chen}.
\end{Remark}

\begin{Lemma}
The \gls{mmse} estimate of the subgroup channel $\boldsymbol{h}^g$ at the
\gls{bs} can be written as~\cite[Sec.~3.2]{2017Bjornsonbook}
\begin{equation}
     \hat{\boldsymbol{h}}^g =
     K_g\, \boldsymbol{R}^g\, \boldsymbol{\Gamma}_{g}^{-1}\,
     \boldsymbol{y}^g ,
     \label{eq:h_gl_est}
\end{equation}
where
\begin{equation}
     \boldsymbol{\Gamma}_{g} =
     \tau_{\mathrm{p}} \sum_{k\in\mathcal{K}_g} q_{gk}\boldsymbol{R}_{k}
     + \sigma_{\mathrm{u}}^2\boldsymbol{I}_{M}.
     \label{eq:gamma_lg}
\end{equation}
The subgroup channel estimate follows a circularly symmetric complex
Gaussian distribution,
$\hat{\boldsymbol{h}}^g \sim
\mathcal{CN}(\boldsymbol{0}_M,
K_g^2 \boldsymbol{R}^g \boldsymbol{\Gamma}_{g}^{-1} \boldsymbol{R}^g)$,
and is uncorrelated with the subgroup channel estimation error
$\tilde{\boldsymbol{h}}^g \sim
\mathcal{CN}(\boldsymbol{0}_M,
\boldsymbol{R}^g -
K_g^2 \boldsymbol{R}^g \boldsymbol{\Gamma}_{g}^{-1} \boldsymbol{R}^g)$.
\end{Lemma}

\begin{IEEEproof}
See Appendix~A.
\end{IEEEproof}

\subsection{MULTICAST TRANSMISSION AND PER-SUBGROUP SPECTRAL EFFICIENCY}

\textcolor{black}{Downlink transmissions are carried out to subgroups of
users rather than to individual users (unicast) or a single multicast
group. All multicast subgroups convey the \emph{same information stream}.
Subgrouping assigns different uplink pilots and downlink precoders to
subsets of users in order to exploit heterogeneity in their spatial
channel statistics, without creating additional services or contents.}

\textcolor{black}{In this setting, the downlink signal received at user
$k$ belonging to subgroup $g$ is given by}
\begin{equation}
      y_{k} 
      = \boldsymbol{h}_{k}^{\mathsf{H}} \mathbf{w}_{g}\varsigma_{g} 
      + \sum_{\substack{c=1 \\ c\neq g}}^{G}
      \boldsymbol{h}_{k}^{\mathsf{H}} \mathbf{w}_{c}\varsigma_{c} 
      + n_k,
      \label{eq:y_k}
\end{equation}
\textcolor{black}{where $n_k \sim \mathcal{CN}(0,\sigma_{\mathrm{d}}^{2})$
denotes the additive white Gaussian noise at user $k$,
$\mathbf{w}_{g}\in\mathbb{C}^{M}$ is the precoding vector used for subgroup
$g$, and $\varsigma_{g}$ is the data symbol intended for all users in that
subgroup, satisfying $\mathbb{E}\{|\varsigma_{g}|^{2}\}=1$ and
$\mathbb{E}\{\varsigma_{g}\varsigma_{c}^{\ast}\}=0$ for all $g\neq c$
(each multicast stream employs a distinct modulation and coding scheme,
rendering symbols intended for different subgroups uncorrelated).}

\textcolor{black}{In~\eqref{eq:y_k}, the first term corresponds to the
desired multicast signal, while the second term accounts for
interference arising from transmissions intended for the other
subgroups. In line with standard massive MIMO operation, users do not
acquire instantaneous downlink \gls{csi} and instead rely on the
statistical expectation
$\mathbb{E}\{\boldsymbol{h}_{k}^{\mathsf{H}}\mathbf{w}_{g}\}$ as a
deterministic equivalent of the effective channel, owing to channel
hardening. Accordingly, the achievable downlink \gls{se} can be
characterized using the well-known hardening bound, obtained by treating
interference and noise as uncorrelated disturbances
\cite{2016Marzetta,2017Bjornsonbook,2018ChenHardening,2020Polegre}.}

\textcolor{black}{The achievable spectral efficiency of user $k$ in subgroup
$g$ is given by}
\begin{equation}
    \xi_{k}
    = \left(1 - \frac{\tau_{\mathrm{p}}}{\tau_{\mathrm{c}}}\right)
      \log_{2}\big(1+\gamma_{k}\big),
    \label{eq:SE_k}
\end{equation}
\textcolor{black}{where the effective signal-to-interference-plus-noise
ratio (SINR) $\gamma_{k}$ is expressed as}
\begin{equation}
        \gamma_{k}
        = \frac{ \big|\mathbb{E}\{\boldsymbol{h}_{k}^{\mathsf{H}}\mathbf{w}_{g}\}\big|^{2} }
        {\displaystyle 
            \sum_{c=1}^{G} \mathbb{E}\!\left\{ \left|\boldsymbol{h}_{k}^{\mathsf{H}}\mathbf{w}_{c}\right|^{2} \right\}
            - \big|\mathbb{E}\{\boldsymbol{h}_{k}^{\mathsf{H}}\mathbf{w}_{g}\}\big|^{2}
            + \sigma_{\mathrm{d}}^{2} }.
    \label{eq:SINR_k}
\end{equation}
{\color{black}This expression follows the standard hardening bound commonly used in massive MIMO systems.}

\textcolor{black}{Since all users in subgroup $g$ decode the same multicast
stream, the spectral efficiency of that subgroup is constrained by the
user experiencing the lowest individual spectral efficiency. It is
therefore given by}
\begin{equation}
    \Xi_{g} = \min_{k\in\mathcal{K}_{g}} \xi_{k}.
    \label{eq:SE_g}
\end{equation}

\subsection{SUBGROUP PRECODING}

\textcolor{black}{\Gls{cb} is adopted as the downlink precoding scheme in the considered multicast massive MIMO system. This precoding scheme is adopted as a baseline to evaluate the \gls{se} under different multicast subgrouping configurations. Each multicast subgroup $g$ is served by a dedicated precoding vector $\mathbf{w}_g$ constructed from the corresponding subgroup channel estimate. The transmit signal intended for subgroup $g$ is scaled by a non-negative power coefficient $\rho_g$, which denotes the downlink transmit power allocated to that subgroup.}

The \gls{cb} precoder used to transmit the multicast signal to subgroup $g$ is defined as
\begin{equation}
    \mathbf{w}_{g}
    = \sqrt{\rho_{g}}\,
    \frac{\hat{\boldsymbol{h}}^{g}}
    {\sqrt{\mathbb{E}\{\|\hat{\boldsymbol{h}}^{g}\|^2\}}},
    \label{eq:MR}
\end{equation}

\section{LSTM-ASSISTED MULTICAST SUBGROUPING}
\label{sec:LSTM-assisted}

{\color{black}In contrast to heuristic or exhaustive subgroup selection methods, we propose a data-driven framework to infer the number of multicast subgroups directly from spatial channel statistics.}

{\color{black}For clarity, the proposed LSTM-assisted multicast subgrouping method can be summarized in three main steps: \emph{(i)} per-snapshot PCA-based preprocessing to obtain a fixed-length representation of the user covariance matrices, \emph{(ii)} sequential processing using an LSTM encoder to capture latent spatial dependencies among users, and \emph{(iii)} multi-output classification through a dense decoder to estimate the number of multicast subgroups.}

\textcolor{black}{The \gls{lstm}-assisted model aims at estimating the most likely number of multicast subgroups for a given snapshot of multicast users in a \gls{mamimo} scenario. Unlike prior learning-based approaches that rely on fixed-size feature representations or scenario-dependent inputs, the proposed framework explicitly leverages second-order spatial channel statistics and processes variable-length user sets through a sequential encoding mechanism. This enables the model to infer the number of multicast subgroups in a manner that reflects the underlying spatial structure of the user distribution, rather than fitting a direct mapping between input features and a predefined subgroup count. For each input data example $i$, we define a matrix $\boldsymbol{S}_i \in \mathbb{R}^{K_i\times(2MM)}$ that contains the real-valued representations of the spatial covariance matrices of the $K_i$ multicast users, obtained by stacking the real and imaginary parts. This matrix constitutes the input to the \gls{lstm}-assisted multicast subgrouping framework and enables the BS to determine the corresponding number of subgroups $G$.}

\textcolor{black}{The \gls{lstm}-assisted model first estimates the optimal
number of multicast subgroups $G$, and subsequently determines the
collection of subgroups $\{\mathcal{K}_g\}_{g=1}^{G}$ by assigning each
user to the subgroup whose \emph{collective} fingerprint (i.e., subgroup
centroid) is closest to its own. The underlying rationale is that pairs
of users whose fingerprints yield values of
$f_{\mathrm{d}}(\boldsymbol{R}_k,\boldsymbol{R}_{k'})$ close to $1$ are
expected to be spatially proximate and to exhibit similar large-scale
channel statistics. Grouping such users within the same multicast subgroup improves channel estimation quality and the effectiveness of the resulting precoders~\cite{2018Riera}.}

\textcolor{black}{The following subsections describe the neural architecture, input preprocessing stage, and training procedure used to implement this strategy.}

\textcolor{black}{To quantify the similarity between users' spatial statistics, we adopt the normalized trace-based metric
\begin{equation}
f_{\mathrm{d}}(\boldsymbol{R}_k, \boldsymbol{R}_{k'}) =
\frac{\mathrm{tr}\!\left(\boldsymbol{R}_{k}\boldsymbol{R}_{k'}\right)}
{\mathrm{tr}\!\left(\boldsymbol{R}_{k}\right)
 \mathrm{tr}\!\left(\boldsymbol{R}_{k'}\right)},
\label{eq:metric} 
\end{equation}
where $R_k$ and $R_{k'}$ denote the spatial covariance matrices of users $k$ and $k'$, respectively.
The metric $f_{\mathrm{d}}(\boldsymbol{R}_k, \boldsymbol{R}_{k'})$ takes values in $[0,1]$ and measures the alignment of the dominant eigenspaces of the two covariance matrices. Values close to one indicate strongly aligned spatial signatures, while values close to zero correspond to
nearly orthogonal channel statistics. This property makes $f_{\mathrm{d}}(\cdot)$ particularly suitable for identifying statistically homogeneous subsets of users for multicast subgrouping.}

\subsection{LSTM-ASSISTED MODEL ARCHITECTURE} \label{subsec:architecture}

\begin{figure}[!t]
\centering
\includegraphics[width=0.95\linewidth]{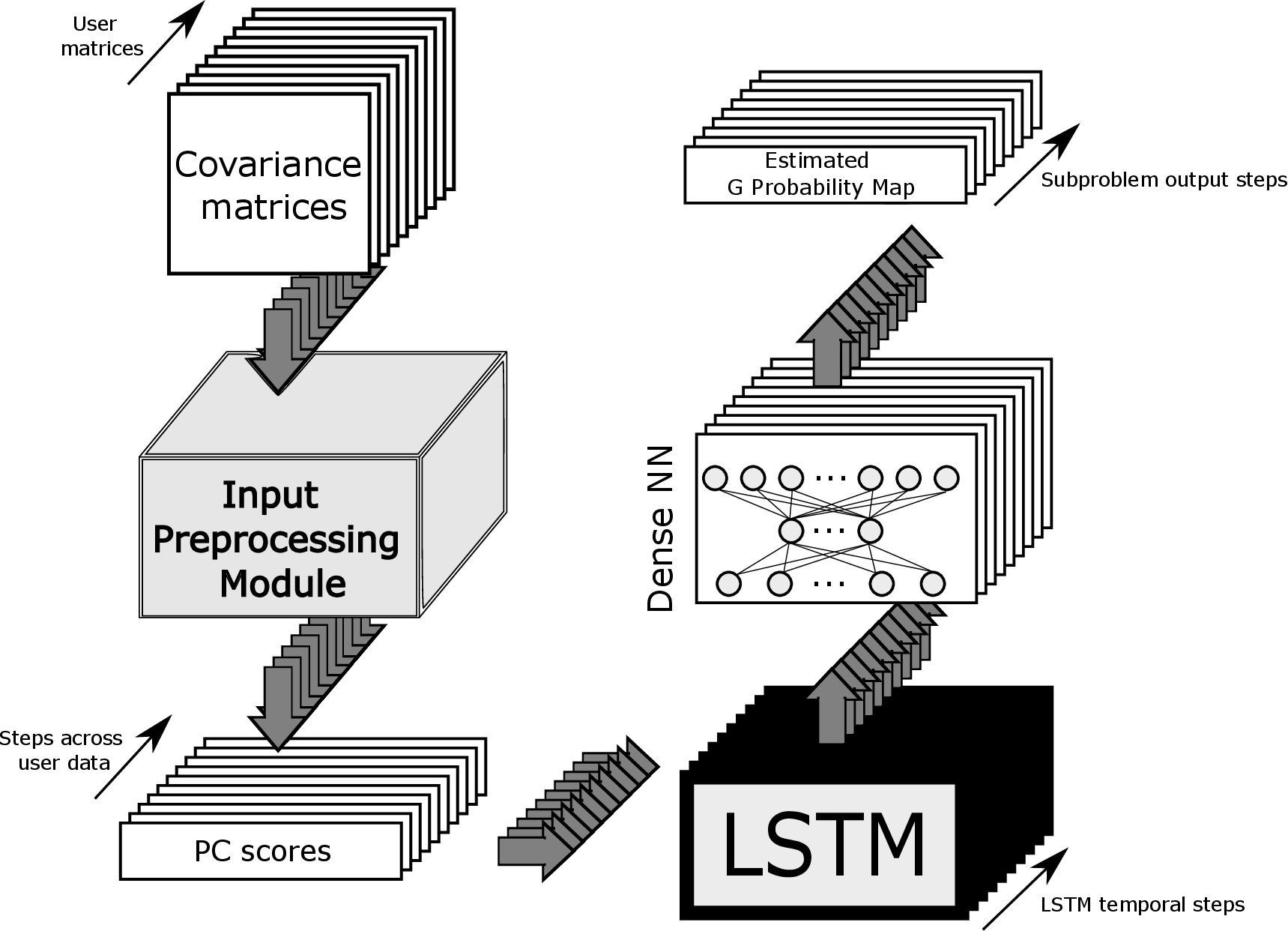}
\caption{\textcolor{black}{Architecture of the LSTM-assisted multicast subgrouping model. User covariance matrices are first processed by the input preprocessing module, where a per-snapshot PCA transformation produces the PC-score representation. The resulting user sequence is processed by the LSTM layer, and the dense neural network maps the LSTM outputs into probability distributions over the possible number of multicast subgroups.}}
\label{Fig:LSTM_model}
\end{figure}

\textcolor{black}{Let the number of multicast users $K_i$ vary across input snapshots, and let $K_{\text{max}}$ denote the maximum number of multicast users considered in the system. Fig.~\ref{Fig:LSTM_model} provides a global view of the proposed LSTM-assisted multicast subgrouping architecture. The use of an LSTM architecture is motivated by its ability to process variable-length sequences and capture dependencies across ordered user representations, which is critical for modeling the spatial structure of multicast user sets. The processing chain is composed of three main blocks: an input preprocessing module, which converts the user covariance matrices into a fixed-size PC-score representation; an LSTM layer, which sequentially processes the user descriptors and captures the latent spatial structure of the multicast user set; and a dense neural network, which maps the LSTM outputs into probability distributions over the possible number of multicast subgroups.}

\textcolor{black}{Since the number of multicast users may vary between different snapshots, the proposed architecture is designed around the maximum problem dimension $K_{\text{max}}$, allowing the framework to operate under heterogeneous multicast user configurations while maintaining a consistent internal representation. In particular, both the dimensionality of the PCA-based representation and the memory size of the LSTM layer are linked to $K_{\text{max}}$, enabling the model to process variable-size user distributions without introducing artificial representational bottlenecks. This design preserves scalability without increasing computational complexity as the number of multicast users grows.}

\textcolor{black}{Fig.~\ref{Fig:PCA} further details the PCA-based preprocessing stage, whereas Fig.~\ref{Fig:LSTM_Training_structure} illustrates how the intermediate LSTM outputs are used during the multi-output training process. Together, these figures provide a complete and coherent view of the proposed framework, from input representation to sequential encoding and final subgroup-number estimation.}

\subsubsection{INPUT PREPROCESSING MODULE} \label{subsec:pca}

The LSTM-based model utilizes \gls{pca} as a dimensionality reduction technique to handle the large size of the spatial covariance matrices associated with each data example. Each data example $i$ contains multiple covariance matrices, corresponding to the $K_i$ users in the scenario. Unlike the conventional approach, where \gls{pca} is applied once across the entire dataset, here, a separate \gls{pca} is performed for each new set of covariance matrices $\boldsymbol{S}_i$. This approach results in a different transformation for each scenario, tailored specifically to the distribution of multicast users within that example $i$.

Every new set of spatial covariance matrices $\boldsymbol{S}_i$ introduced to the model requires performing a separate PCA and applying a different transformation. Since each scenario $i$ (snapshot of users' distribution) creates its own dataset of spatial covariance matrices, this approach allows us to capture the most relevant principal components specific to that particular distribution of multicast users. The main goal is to reduce the size of the input data $\boldsymbol{S}_i \in \mathbb{R}^{K_i \times (2MM)}$ (i.e., the covariance matrices of $K_i$ users), to a fixed set of principal components denoted by $\boldsymbol{P}_i \in \mathbb{R}^{K_i \times K_{\text{max}}}$, which serve as input to the LSTM layer. Fig. \ref{Fig:PCA} illustrates the input dimensionality reduction using PCA for each data example which is achieved in three steps: a) the PCA output $\boldsymbol{V}_i \in \mathbb{R}^{K_i \times (2MM)}$ presents the principal components for the $K_i$ users maintaining the dimensions of the input $\boldsymbol{S}_i$, b) the projection of the principal components on the covariance matrices $\langle\boldsymbol{S}_i,\boldsymbol{V}_i\rangle$ results in the PC (principal components) scores matrix $\boldsymbol{P'_i} \in \mathbb{R}^{K_i \times K_i}$, and c) if the number of users $K_i < K_{\text{max}}$, the resulting PCA output matrix contains only $K_i$ principal components for the $K_i$ multicast users, thus, the $K_i$ vectors corresponding to the $K_i$ multicast users are padded with zeros to maintain a consistent input size for the LSTM, $\boldsymbol{P_i} \in \mathbb{R}^{K_i \times K_{\text{max}}}$, where the number of multicast users $K_i$ is variable but the number of principal components $K_{\text{max}}$ is fixed. This process ensures that the LSTM-based model can dynamically adjust to different numbers of users in the cell, making it robust in learning from varied spatial structures and clusters of users' formations across different cases.

\begin{figure}[!t]
\centering
\includegraphics[width=0.95\linewidth]{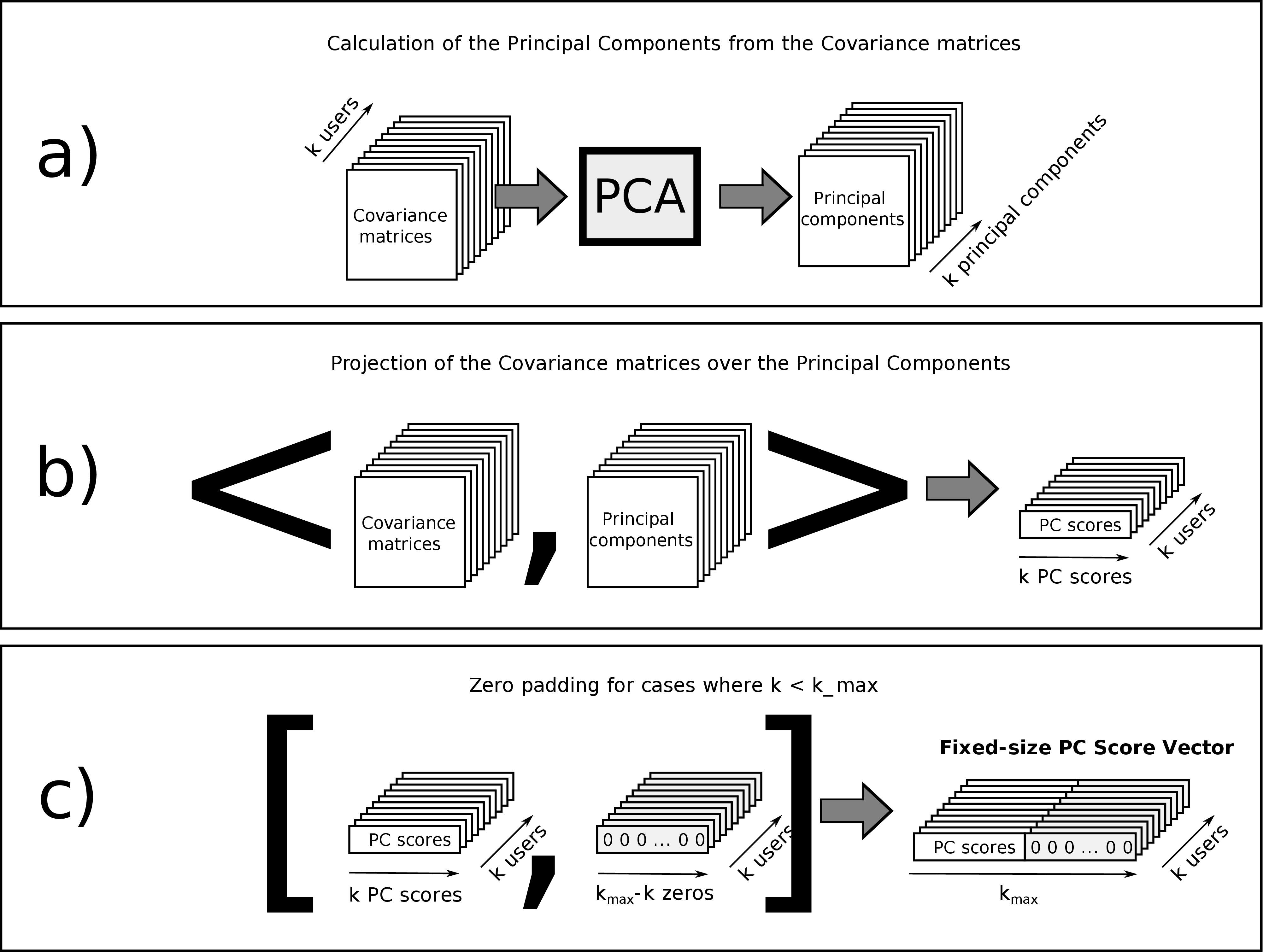}
\caption{Input dimensionality reduction using PCA for each data example in three steps: a) the PCA of the covariance matrices, b) the projection of the principal components on the covariance matrices results in the PC scores, and c) the $K_i$ multicast
users are padded with zeros to maintain a consistent input size of $K_{\text{max}}$ principal components for the LSTM.}
\label{Fig:PCA}
\end{figure}

\textcolor{black}{For a generic snapshot with $K_i \leq K_{\text{max}}$ users, this procedure yields a matrix $\boldsymbol{P}_i \in \mathbb{R}^{K_i \times K_{\text{max}}}$. During the multi-output training process described in the following subsection, snapshots with $K_{\text{max}}$ users are used so that the LSTM can learn subgroup-number estimation problems for all intermediate user subset sizes, from $2$ to $K_{\text{max}}$.}

\subsubsection{LSTM LAYER}
The output data of the preprocessing network $\boldsymbol{P}_i$ will be the input of the following layer in the model, an LSTM layer with a memory size of $K_{\text{max}}$ units. This layer processes user-related information sequentially, storing the relevant elements of each user in its memory in terms of the spatial structure that forms the clusters of multicast users in the cell. Each time step in the LSTM corresponds to processing the spatial covariance matrix information of a new user $k$, and the LSTM's ability to retain temporal information allows the model to preserve important details about the users' spatial arrangement as it processes each input.

\subsubsection{DENSE NEURAL NETWORK}
The output of the LSTM layer is the input to a dense neural network, which is composed of two hidden layers followed by an output layer. Table \ref{table:Dense_NN} presents the activation functions and the number of units of each layer.

The third and final layer using \texttt{softmax} activation, is designed to generate a probability distribution over the number of multicast subgroups $G_i$, which can range from $1$ to $K_{\text{max}}$. The purpose of this dense network is to take the information processed by the LSTM and estimate the value of $G_i$ that maximizes the probability of correctly representing the spatial cluster distribution of the multicast users.

\begin{table}[t]
\centering
\caption{Dense neural network activation functions and size.}
\begin{tabular}{|l|l|l|}
\hline
\textbf{Layer} & \textbf{Activation function} & \textbf{Size (units)}\\ \hline
First layer & \texttt{tanh}  & $K_{\text{max}}/2$\\ \hline
Second layer & \texttt{sigmoid} & $3K_{\text{max}}/4$\\ \hline
Third layer & \texttt{softmax} & $K_{\text{max}}$\\ \hline
\end{tabular}
\label{table:Dense_NN}
\end{table}

\subsection{TRAINING PROCESS} \label{subsec:training_process}

The training process for the LSTM-based subgrouping model begins with preprocessing the training data through PCA.  As mentioned above, this step reduces the spatial covariance matrix size of the example data from  $\boldsymbol{S}_i \in \mathbb{R}^{K_i \times (2MM)}$ to $\boldsymbol{P}_i \in \mathbb{R}^{K_i \times K_{\text{max}}}$. We design the training dataset to always include $K_{\text{max}}$ users, including scenarios with different number of spatial clusters of users and also several spatial cluster radii. Thus, the input matrix that serves as the data fed into the LSTM layer $\boldsymbol{P}_i \in \mathbb{R}^{K_{\text{max}} \times K_{\text{max}}}$ is created for each data example $i$ (snapshot). Note that this input data contains $K_{\text{max}}$ principal component values for each of the $K_{\text{max}}$ multicast users.

As presented in Fig.~\ref{Fig:LSTM_Training_structure}, the LSTM layer is configured to produce an intermediate hidden representation at each time step, corresponding to the sequential processing of each user's data. Although the first user's output is discarded, the subsequent $K_{\text{max}} - 1$ outputs are passed through the dense neural network, which estimates the number of multicast subgroups $G_i$ as multicast users are progressively added from $2$ to $K_{\text{max}}$.

\textcolor{black}{
This training strategy allows the framework to learn subgrouping problems of progressively increasing size within a single forward pass. As each new multicast user is sequentially introduced into the LSTM memory, the decoder network is applied to all valid intermediate outputs, enabling the model to estimate the number of multicast subgroups for user subsets ranging from $2$ to $K_{\text{max}}$. Consequently, the proposed architecture jointly learns the evolution of the subgroup structure as additional users are incorporated into the multicast scenario, improving its ability to generalize across heterogeneous user distributions and variable problem dimensions.
}

\subsubsection{OUTPUT AND LABELING}
To create the target outputs for training, we leverage the fact that each synthetic user $k$ is labeled during the data generation process, indicating the spatial cluster $g$ to which it belongs. This labeling allows us to easily compute the true number of spatial clusters in each user distribution $i$. These true cluster values are encoded using one-hot encoding, making them compatible with the \texttt{softmax} output of the dense network. Therefore, the model is trained to optimize across $K_{\text{max}}-1$ outputs, learning to predict the number of spatial clusters $G_i$, for problems involving scenarios with users distributions varying from $2$ to $K_{\text{max}}$ users.

The resulting one-hot vectors define the supervised targets used in the multi-output training process. The corresponding loss function and optimization procedure are described next.

{\color{black}
\subsubsection{TRAINING CONFIGURATION}
The proposed LSTM-assisted multicast subgrouping framework is trained as a multi-output classification problem. For each input snapshot $i$, the LSTM sequentially processes the multicast users and generates intermediate outputs associated with subgrouping problems of increasing size. The first LSTM output is discarded, while the remaining outputs, corresponding to multicast user subsets ranging from $2$ to $K_{\text{max}}$, are fed into the dense decoder network. This strategy allows the model to jointly learn subgrouping structures for multiple multicast user configurations within a single forward pass.

The target output associated with each valid LSTM step is represented using one-hot encoding, where the active position corresponds to the true number of spatial clusters associated with the considered multicast user subset. Let $\boldsymbol{y}_{i,t}\in\{0,1\}^{K_{\text{max}}}$ denote the target vector for snapshot $i$ and LSTM output step $t$, and let $\widehat{\boldsymbol{y}}_{i,t}\in[0,1]^{K_{\text{max}}}$ be the corresponding softmax prediction generated by the dense network. The training process minimizes the categorical cross-entropy loss
\begin{equation}
\mathcal{L}_{\mathrm{CE}}
=
-\frac{1}{N\left(K_{\text{max}}-1\right)}
\sum_{i=1}^{N}
\sum_{t=2}^{K_{\text{max}}}
\sum_{g=1}^{K_{\text{max}}}
y_{i,t,g}\log\!\left(\widehat{y}_{i,t,g}\right),
\label{eq:lstm_loss}
\end{equation}
where $N$ denotes the number of training snapshots, $t$ indexes the valid LSTM output steps associated with multicast user subsets of sizes from $2$ to $K_{\text{max}}$, $g$ indexes the possible number of multicast subgroups, $y_{i,t,g}$ is the one-hot encoded target, and $\widehat{y}_{i,t,g}$ is the probability predicted by the softmax output of the dense decoder.

The network parameters are optimized using the RMSprop algorithm. Let $\boldsymbol{\theta}$ denote the set of trainable parameters of the LSTM and dense layers. At each training iteration, the parameters are updated according to
\begin{equation}
\boldsymbol{\theta}^{(n+1)}
=
\boldsymbol{\theta}^{(n)}
-
\eta
\frac{\nabla_{\boldsymbol{\theta}}\mathcal{L}_{\mathrm{CE}}}
{\sqrt{\boldsymbol{v}^{(n+1)}}+\epsilon},
\label{eq:rmsprop}
\end{equation}
where $\eta$ is the learning rate, $\epsilon$ is a numerical stability constant, and $\boldsymbol{v}^{(n+1)}$ is the exponentially weighted moving average of the squared gradients.

Training is performed for $10$ epochs using mini-batches containing $32$ snapshots. During each epoch, the order of multicast users within every snapshot is randomly shuffled before being processed by the LSTM. This randomization prevents the recurrent model from learning position-dependent artifacts and improves the ability of the framework to generalize across heterogeneous multicast user distributions and spatial configurations. Since LSTM networks process inputs sequentially, they may be sensitive to the ordering of the input sequence. To mitigate this effect, the order of multicast users is randomly shuffled at each epoch during training. This prevents the model from learning position-dependent artifacts and improves generalization. At inference time, the impact of input ordering is observed to be negligible due to the statistical nature of the covariance-based representation.

Fig.~\ref{Fig:LSTM_Training_structure} illustrates the complete training structure of the proposed LSTM-assisted multicast subgrouping framework.

}

\subsection{INPUT DATA REPLICATION ENHANCEMENT} \label{subsec:enhancement}

{\color{black}This enhancement is introduced as a practical mechanism to improve the stability of the sequential encoding process when handling variable-size user sets. The impact of this design choice is evaluated empirically in Section~IV.}

The initial configuration of the LSTM-assisted subgrouping model is designed to estimate the number of subgroups $G_i$ for a variable number of multicast users $K_i$ (ranging from 2 to $K_{\text{max}}$). In this setup, for a given number of users $K_i<K_{\text{max}}$, the LSTM layer processes each user $k$ individually, and the output is passed to the dense layer to predict $G_i$. Although this approach works correctly, an enhancement has been developed to significantly improve the model's performance.

By replicating the input data to always provide $K_{\text{max}}$ users, even if $K_i<K_{\text{max}}$, the system refines its ability to estimate the spatial cluster structure more effectively. Although this strategy introduces repeated information, it enhances the capacity of the LSTM layer to store and recall relevant details, allowing for more accurate spatial cluster estimation. The repetition acts as a mechanism for the LSTM-based model to refine its internal state, effectively reinforcing the learned representation of the user configuration. This approach leverages the LSTM's ability to retain and process information over multiple time steps. For instance, if there are $K_i<K_{\text{max}}$ users, the data for those $K_i$ users is replicated to fill all $K_{\text{max}}$ time steps, which helps the model focus on refining the cluster structure for a more precise estimation of the optimal number of subgroups $G_i$. The system, designed to handle variable $K_i$ sizes, benefits from this refinement process, improving its robustness and accuracy across different users' and spatial clusters' distributions.

In summary, this data replication technique helps the LSTM-based model refine its memory of the relevant spatial cluster information, resulting in an enhanced performance compared to the initial method, where the number of LSTM steps was directly tied to the actual number of users $K_i$. Note that this enhancement fixes the size of the LSTM layer to consist of $K_{\text{max}}$ units.

\begin{figure}[!t]
\centering
\includegraphics[width=\linewidth]{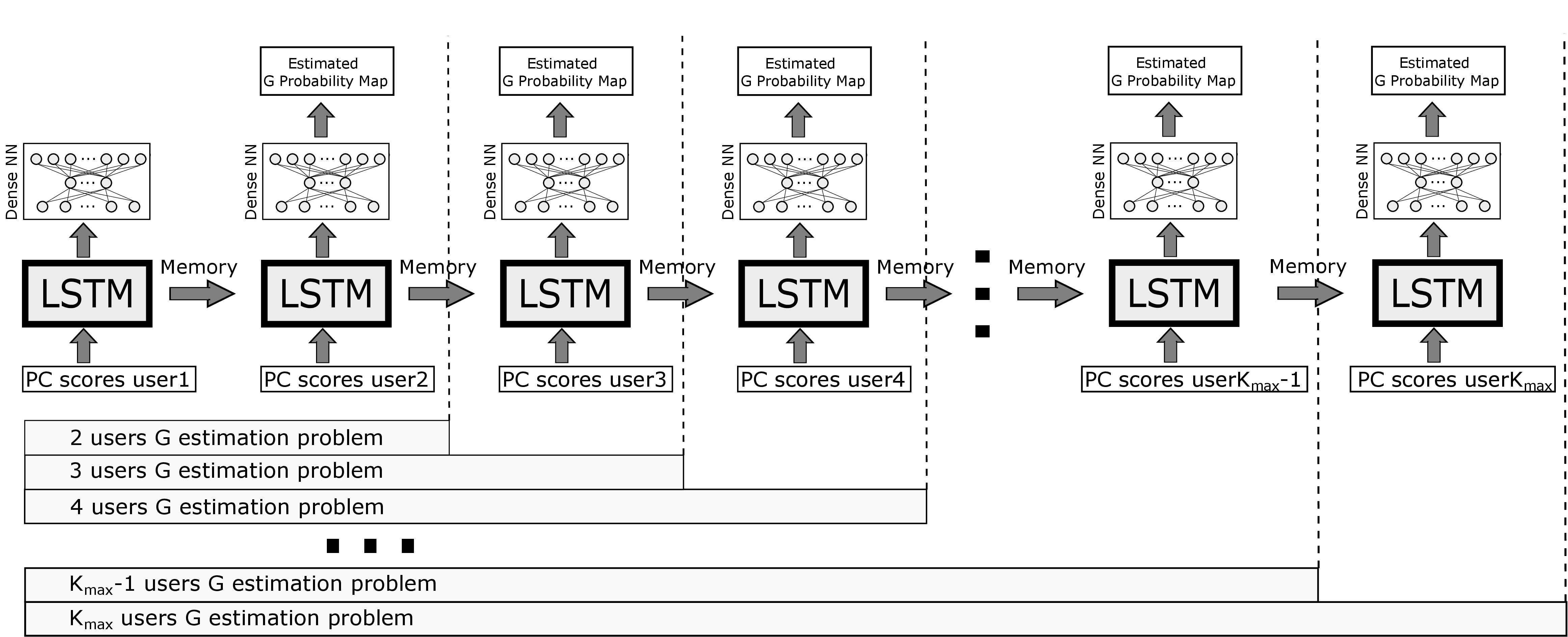}
\caption{Structure of the LSTM-assisted multicast subgrouping training process. The LSTM produces one output per processed user. The first output is discarded, while the remaining outputs are passed through the dense decoder to estimate the number of multicast subgroups for user subsets of increasing size, from $2$ to $K_{\text{max}}$.}
\label{Fig:LSTM_Training_structure}
\end{figure}

\subsection{TRANSFER LEARNING FOR SPECTRAL EFFICIENCY ESTIMATION ENHANCEMENT}
\label{subsec:transfer_learning}

Initially, the LSTM-assisted model is trained to estimate the number of spatial clusters represented in the user distribution. However, the ultimate objective of the multicast subgrouping problem is to determine the number of subgroups $G_i$ that maximizes the sum \gls{se} $\Psi$. To address this objective, we propose a \gls{tl} extension, denoted as LSTM (TL-SE), which leverages the spatial representation learned by the pre-trained LSTM encoder while training a modified dense network to estimate the sum \gls{se} associated with the different feasible subgroup configurations.

Importantly, the subgroup number $G$ is neither provided as an input to the LSTM-based model nor fixed to a predefined value. For a snapshot $i$ containing $K_i$ multicast users, the model receives only the PCA-based representation of their spatial covariance matrices. The LSTM encoder then generates a latent representation of the spatial user distribution, which is processed by the dense output head to produce the vector
\begin{equation}
\widehat{\boldsymbol{\Psi}}_i
=
\left[
\widehat{\Psi}_{i,1},
\widehat{\Psi}_{i,2},
\ldots,
\widehat{\Psi}_{i,K_{\mathrm{max}}}
\right],
\label{eq:tl_se_output}
\end{equation}
where $\widehat{\Psi}_{i,g}$ denotes the predicted normalized sum \gls{se} associated with partitioning the users into $g$ multicast subgroups. Therefore, the output index $g$ represents a candidate value of $G$, rather than a subgroup number supplied externally to the network.

For a scenario containing $K_i$ users, only the subgroup configurations satisfying $1 \leq g \leq K_i$ are feasible. Accordingly, the number of multicast subgroups selected by the LSTM (TL-SE) model is given by
\begin{equation}
\widehat{G}_i
=
\underset{g\in\{1,\ldots,K_i\}}{\argmax}
\;
\widehat{\Psi}_{i,g}.
\label{eq:tl_se_selection}
\end{equation}

Hence, all feasible subgroup counts are evaluated simultaneously through the output layer in a single forward pass, and no default or previously selected value of $G$ is required.

\subsubsection{MODEL ARCHITECTURE}

The LSTM (TL-SE) model retains the PCA-based preprocessing module and the same LSTM structure and weights learned during the initial subgroup-number classification task. During the transfer-learning stage, the pre-trained LSTM weights are frozen, while the dense output network is trained specifically to estimate the normalized sum \gls{se} associated with the candidate subgroup configurations.

Unlike the original LSTM-assisted classifier, whose output units represent the probabilities associated with the possible numbers of spatial clusters, the output units of the LSTM (TL-SE) model have a regression-based interpretation. In particular, the $g$-th output unit estimates the normalized sum \gls{se} obtained when the $K_i$ users are divided into $g$ multicast subgroups.

The modified dense network maintains a three-layer structure with $K_{\mathrm{max}}/2$, $3K_{\mathrm{max}}/2$, and $K_{\mathrm{max}}$ units, respectively. The output layer therefore contains one unit for each possible subgroup count. Its activation function is changed from \texttt{softmax} to \texttt{sigmoid}, since the objective is no longer to generate a probability distribution over mutually exclusive classes, but to predict normalized sum-\gls{se} values for the different subgroup configurations. The sigmoid activation constrains each predicted value $\widehat{\Psi}_{i,g}$ to the same normalized range used for the training targets.

For snapshots with $K_i<K_{\mathrm{max}}$ users, the output units corresponding to $g>K_i$ represent infeasible subgroup configurations and are excluded from the subgroup-selection rule in~\eqref{eq:tl_se_selection}.

\subsubsection{TRANSFER LEARNING APPROACH}

In the original LSTM-assisted model, the intermediate outputs generated while users are sequentially processed can be used to estimate the number of spatial clusters for progressively larger user subsets. In contrast, obtaining a sum-\gls{se} target for each intermediate subset would require independently simulating the complete multicast transmission procedure for every subset size and every possible subgroup count. This would substantially increase the computational cost of dataset generation.

Consequently, the LSTM (TL-SE) problem is formulated as a global snapshot-level estimation task rather than as the multi-output subset-based classification task used in the original model. For each training snapshot $i$, the target output is a vector containing the normalized sum \gls{se} associated with the different feasible numbers of multicast subgroups:
\begin{equation}
\boldsymbol{\Psi}_i
=
\left[
\Psi_{i,1},
\Psi_{i,2},
\ldots,
\Psi_{i,K_i}
\right],
\label{eq:tl_se_target}
\end{equation}
where $\Psi_{i,g}$ is the normalized sum \gls{se} obtained after partitioning the users of snapshot $i$ into $g$ subgroups and evaluating the resulting multicast transmission. Thus, the model learns to predict the dependence of the achievable sum \gls{se} on the subgroup count without requiring $G$ as an input variable.

The dataset used for this task is smaller than the dataset employed to train the original LSTM-assisted classifier. Moreover, the sum-\gls{se} simulations cover only a limited range of user counts and spatial-cluster configurations, rather than the complete range from $2$ to $K_{\mathrm{max}}$ users and from $1$ to $K_{\mathrm{max}}$ spatial clusters. Training the complete architecture exclusively with this reduced dataset would therefore increase the risk of overfitting and limit its ability to generalize to previously unseen user distributions.

To mitigate this limitation, the LSTM encoder pre-trained on the original subgroup-number estimation task is reused as a fixed spatial-feature extractor. This encoder has already learned to represent the latent spatial structure of user distributions from a considerably broader collection of user counts, spatial-cluster numbers, and cluster configurations. The transfer-learning stage therefore only needs to train the dense regression head to map this latent representation to the sum-\gls{se} values associated with the candidate subgroup counts.

By preserving the spatial knowledge encoded by the pre-trained LSTM, the proposed transfer-learning strategy reduces the number of trainable parameters and mitigates overfitting to the smaller sum-\gls{se} dataset. The dense network can consequently focus on learning the relationship between the encoded spatial user structure, the candidate subgroup configurations, and the corresponding sum \gls{se}. At inference time, the model predicts the sum \gls{se} for all feasible values of $G$ and selects the subgroup count according to~\eqref{eq:tl_se_selection}.

{\color{black} \subsection{BASELINE PRECODING AND POWER CONTROL MODEL}
\label{subsec:precoding_power_model}

{\color{black}This subsection describes the baseline signal processing and power control models used exclusively for performance evaluation and not as part of the proposed learning-based subgrouping mechanism.}

\Gls{cb} is adopted as the baseline downlink precoding scheme to evaluate the spectral efficiency under different multicast subgrouping configurations. This choice is motivated by its low computational complexity, robustness to covariance uncertainty, and well-known asymptotic optimality properties in the large-antenna regime. In particular, as the number of \gls{bs} antennas increases, \gls{cb} aligns the transmitted signal with the dominant channel directions, while inter-group and intra-group interference terms vanish in expectation due to channel hardening and favorable propagation.

For finite antenna dimensions, \gls{cb} does not completely suppress co-channel interference among subgroups. Nevertheless, this residual interference is naturally captured by the spectral efficiency evaluation used throughout this work and does not bias the subgrouping selection. Indeed, clustering users with similar second-order channel statistics yields effective channels that are well matched to the corresponding subgroup precoders. As a result, \gls{cb} provides a representative and scalable reference for assessing the impact of multicast subgrouping.

More sophisticated multicast beamforming techniques based on optimization frameworks such as \gls{sdr}, \gls{sca}, or \gls{wmmse} can achieve higher performance under their native assumptions. However, these approaches require full per-user instantaneous \gls{csi} and entail computational complexity that scales cubically with the number of users and antennas, making them unsuitable for large-scale multicast scenarios under practical \gls{csi} acquisition constraints. In contrast, the use of \gls{cb} preserves linear
complexity in both $M$ and $K$, enabling the generation of ground-truth spectral efficiency values for a large number of subgrouping configurations and user realizations.

The multicast signal intended for subgroup $g$ is transmitted using the precoding vector
\begin{equation}
\mathbf{w}_{g}
= \sqrt{\rho_{g}}\,
\frac{\hat{\boldsymbol{h}}^{g}}
{\sqrt{\mathbb{E}\{\|\hat{\boldsymbol{h}}^{g}\|^2\}}},
\label{eq:MR}
\end{equation}
where $\hat{\boldsymbol{h}}^{g}$ denotes the estimated subgroup channel and $\rho_{g} \geq 0$ is the downlink transmit power allocated to subgroup $g$. The power coefficients $\{\rho_g\}$ are determined through scalable power control procedures described later and are used solely to evaluate the resulting spectral efficiency associated with each subgrouping configuration.

The UL pilot powers $\{q_{gk}\}$ determine the estimation quality, directly impacting the effective DL \glspl{sinr}; conversely, the DL subgroup powers $\{\rho_g\}$ set operating margins and thus the sensitivity to estimation errors. {\color{black}The UL/DL problems are therefore coupled. For clarity, we denote by $q_{gk}$ the uplink pilot transmit powers and by $\rho_g$ the downlink transmit powers associated with subgroup $g$.} We adopt the alternating optimization routine analyzed in our previous work \cite{2022delaFuente}: initialize $\{q_{gk}^{(0)}\}$ and $\{\rho_g^{(0)}\}$ (e.g., pathloss-aware pilots and fractional DL power from large-scale statistics); then iterate (i) \emph{UL step}: update $\{q_{gk}\}$ to improve a per-subgroup \gls{mmf} surrogate and refresh estimates/statistics; (ii) \emph{DL step}: update $\{\rho_g\}$ via an \gls{mmf} program (e.g., \gls{sca} with bisection on the \gls{sinr} target) under $\sum_g \rho_g \le P_{\mathrm{dl}}$; stop when the relative improvement falls below $\epsilon$ or at $T_{\max}$. This coordination preserves scalability and avoids duplicating algorithmic details already provided in \cite{2022delaFuente}. In the following, we formulate both the \emph{inter-subgroup} and \emph{intra-subgroup} power control problems.

\subsubsection{INTER-SUBGROUP MAX-MIN FAIR POWER CONTROL}

The inter-subgroup \gls{mmf} power allocation under a downlink sum-power constraint is formulated as
\begin{equation}
\begin{aligned}
\mathcal{P}_1:\quad 
& \underset{\{\rho_g \ge 0\}}{\text{maximize}} \;\;
\underset{g \in \mathcal{G}}{\min}\; \Xi_{g} \\
& \text{subject to} \;\; \sum_{g=1}^{G} \rho_{g} \le P_{\mathrm{dl}},
\end{aligned}
\label{eq:MMF_problem}
\end{equation}
where $P_{\mathrm{dl}}$ denotes the downlink transmit-power budget at the
\gls{bs} per coherence block.

Problem~\eqref{eq:MMF_problem} is nonconvex in the variables $\{\rho_g\}$. However, it can be reformulated in epigraph form by introducing an auxiliary variable, denoted by $\Gamma$, which serves as a quality-of-service constraint. The globally optimal solution can then be obtained in polynomial time by performing a bisection search over $\Gamma$ and solving, at each iteration, a convex feasibility problem. Implementation details of the \gls{mmf} bisection routine are provided in~\cite[Algorithm~2]{2022delaFuente}.

\subsubsection{INTRA-SUBGROUP MAX-MIN FAIR POWER CONTROL}

To compute a near-optimal uplink pilot-power allocation that maximizes
the minimum spectral efficiency within each multicast subgroup, we
adopt a heuristic routine based on the approach described
in~\cite[Algorithm~1]{2022delaFuente}. The procedure is initialized using
the sub-optimal pilot-power setting proposed in~\cite{2018SadeghiTWC1},
given by
\begin{equation}
q_{gk} = \frac{1+\beta_{k}\rho_g}{\beta_{k}^2}\,\Upsilon_g,
\quad \forall\, g,\; k \in \mathcal{K}_g,
\label{eq:q_uncorr}
\end{equation}
where
\begin{equation}
\Upsilon_g = \min_{k \in \mathcal{K}_g}
\frac{P_{\mathrm{ul}}\beta_{k}^2}{1+\beta_{k}\rho_g},
\end{equation}
and $P_{\mathrm{ul}}$ denotes the uplink pilot transmit-power budget per user.

For each multicast subgroup, the algorithm iteratively updates the uplink pilot-power coefficients $\{q_{gk}\}$ to increase the minimum per-user SINR. At each iteration, a suboptimal target value
$q_{gk}^{\star}$ is computed for each user $k \in \mathcal{K}_g$ so as to equalize the SINR across the subgroup. The pilot powers are then updated according to the normalization rule
\begin{equation}
q_{gk} \leftarrow \mu_{k}\,
\frac{q_{gk}^{\star}}{\max\limits_{k \in \mathcal{K}_g} q_{gk}^{\star}},
\end{equation}
where $\mu_k$ denotes a step-size parameter.

Since the uplink pilot-power allocation and the downlink data-power
allocation are inherently coupled, the uplink optimization adopts an
a priori downlink power profile derived from large-scale fading
statistics to decouple the updates. Specifically, the downlink power
allocated to subgroup $g$ is given by
\begin{equation}
\rho_g = P_{\mathrm{dl}}\,
\frac{\beta_g^{\nu}}{\sum_{c=1}^{G} \beta_c^{\nu}},
\label{eq:p_g}
\end{equation}
where
\begin{equation}
\beta_g = \frac{1}{M}\sum_{k \in \mathcal{K}_g}
\Big[\mathrm{tr}(\boldsymbol{R}_{k})
- \mathrm{tr}\!\left(\boldsymbol{R}^g
- K_g^2 \boldsymbol{R}^g \boldsymbol{\Gamma}_{g}^{-1}
\boldsymbol{R}^g\right)\Big],
\end{equation}
and the tuning parameter $\nu \in [-1,1]$ controls the fairness of the
power allocation across subgroups.

After defining the inter- and intra-subgroup power control strategies, the performance metric of interest is the sum spectral efficiency of the multicast transmission, expressed as
\begin{equation}
\Psi = \sum_{g=1}^{G} K_g \Xi_g.
\label{eq:SumSE}
\end{equation}
Since both the uplink pilot powers and the downlink data powers are
optimized according to a max--min fairness criterion, the resulting
allocation equalizes the spectral efficiency across all multicast
users. Consequently, the sum spectral efficiency simplifies to
\begin{equation}
\Psi = K\,\overline{\xi},
\label{eq:SumSE2}
\end{equation}
where $\overline{\xi}$ denotes the common per-user spectral efficiency
for all $k \in \mathcal{K}$.
}

\subsection{DATASETS FOR TRAINING AND EVALUATION} \label{subsec:datasets}

{\color{black}This subsection summarizes the datasets used to train and evaluate the proposed framework, highlighting the key characteristics relevant to learning-based multicast subgrouping.}

The performance of the proposed LSTM-based subgrouping model has been evaluated using two distinct datasets: one for training and another for testing. Both datasets were designed to simulate realistic \gls{mamimo} multicasting scenarios, ensuring that the model could generalize well to various conditions and configurations found in practical systems.

\subsubsection{TRAINING DATA FOR THE ORIGINAL MODEL}

For the training phase, we generated synthetic datasets reflecting the spatial distribution and channel covariance matrices characterizing the multicast users in a single-cell \gls{mamimo} system. This dataset was help the model learn the optimal number of subgroups $G$, under different users' spatial distributions.

The data generation process involved varying the number of spatial clusters of multicast users and the number of users deployed in each cluster, simulating the associated channel covariance matrices for each user within the spatial clusters. Specifically, three different spatial cluster radii have been considered ($2.5$ m, $5$ m, and $15$ m) to reflect environments with varying spatial correlation and user densities. The number of spatial clusters of users ranged from $1$ to $K_{\text{max}}$.

For each combination of spatial cluster radio and number of spatial clusters, we have performed $150$ independent realizations, resulting in a comprehensive dataset consisting of $3 \times 128 \times 150 = 57600$ simulations.
Each simulation $i$ provides input data in the form of $\boldsymbol{S}_i$ that consists of $K_i$ channel covariance matrices $\boldsymbol{R}_k$, along with the ground truth number of spatial clusters deployed $G_{\text{true}}$. This extensive dataset allowed the model to learn from diverse channel conditions and user spatial distributions, making it robust and generalizable.

\subsubsection{TRAINING DATA OF TRANSFER LEARNING FOR SPECTRAL EFFICIENCY MODEL}
To train the enhanced model, it was necessary to generate a new training dataset that includes sum SE $\Psi$ values, similar to the structure of the test dataset. We have created $4680$ cases that reflect the same diversity in terms of spatial cluster radii ($2.5$ m, $5$ m, and $15$ m) and the number of users ($40$, $60$, $80$, and $100$ users) and spatial clusters ($1$, $2$, $4$, $5$, $24$, $30$, $60$, and $120$ spatial clusters). Each case is labeled with the sum SE $\Psi$ achieved using different numbers of multicast subgroups $G$, enabling the model to learn the relationship between the spatial clustering structure and the SE.

The target for training the enhanced model is the normalized sum SE value for each possible $G$. Specifically, for a dataset of users $i$, the sum SE $\Psi$ values are normalized as follows:
\begin{equation}
        \!\Theta_{i} \!=\! \frac{\Psi_i(G) - \Psi_i^{\text{min}}}{\Psi_i^{\text{max}} - \Psi_i^{\text{min}}},
\label{eq:SE_norm}        
\end{equation}
where $\Psi_i(G)$ is the sum SE using $G$ multicast subgroups, and $\Psi_i^{\text{min}}$ and $\Psi_i^{\text{max}}$ are the minimum and maximum sum SE values for that set of users $i$. This normalization ensures that the $\Theta_{i}$ is scaled between $0$ and $1$, making it compatible with the sigmoid activation function used in the final layer of the dense network.

During the \gls{tl} stage, the pre-trained \gls{lstm} encoder is frozen, and only the dense output head is fine-tuned. The model is trained to predict the normalized sum \gls{se} $\Theta_i$ in~(\ref{eq:SE_norm}). Since this is a continuous regression task, we adopt the mean squared error (MSE) as the loss function:
\begin{equation}
    \mathcal{L}_{\mathrm{MSE}}
    = \frac{1}{N}\sum_{i=1}^{N}
    \left( \widehat{\Theta}_{i} - \Theta_{i} \right)^{2},
\end{equation}
where $\widehat{\Theta}_{i}$ denotes the model output. The choice of MSE stabilizes training across different SE ranges and provides smooth gradients for fine-tuning the dense layers.

\subsubsection{TEST DATA}

The test dataset, consisting of $480$ simulation files, has been designed to provide a comprehensive evaluation of the LSTM-based subgrouping method. Unlike the training dataset, the test simulations are more computationally expensive, as they involve calculating the SE for every possible value of $G$ (the number of subgroups). This allows for a deeper assessment of the model's performance in estimating $G$ and its impact on system performance.

Each simulated test data involved configurations with varying numbers of users (ranging from $K_i=40$ to $K_i=120$ users per snapshot) and different users per spatial cluster ratios (e.g., for $K_i=40$ we have $1$ user per cluster and $40$ spatial clusters, $2$ users per cluster and $20$ spatial clusters, $4$ users per cluster and $10$ spatial clusters, ...). The same cluster radii of $2.5$ m, $5$ m, and $15$ m have been employed, ensuring that the test data captured different spatial densities. We have generated $4$ independent realizations for each configuration.

\subsubsection{COMPLEXITY IN SPECTRAL EFFICIENCY CALCULATION}
The key distinction between the test and training datasets lies in the sum SE calculation. In the test dataset, the sum SE is computed for each possible value of $G$, adding significant computational complexity. This additional step is essential for evaluating the full algorithm, as it allows for a comparison of the predicted $G$ with the sum SE obtained from different subgrouping methods.

The test dataset enables a thorough analysis of how well the complete process performs in maximizing the sum SE. The complete process consists of two steps: first, the estimation of the number of subgroups $G$, and second, the application of the K-means algorithm applying the metric defined in~\eqref{eq:metric}.

\section{NUMERICAL RESULTS}
\label{sec:results}
We consider a \gls{mamimo} cell where a single \gls{bs} equipped with a \gls{ula} of $M\!=\!64$ antennas delivers a multicast service over a coverage area with $200$~m radius. 
Within this region, users are placed in spatial clusters of $5$~m radius at uniformly random locations, and their nominal angles of arrival (AoAs) with respect to the \gls{ula} are computed accordingly. 
These positions and AoAs determine the spatial correlation among user channel vectors, yielding a network snapshot for evaluation. 
Table~\ref{Tab:sim_param} summarizes the simulation configuration adopted to generate the training and validation datasets. 
The path loss is modeled as $-32.4 \!-\! 20\log_{10}(f) \!-\! 37.6\log_{10}(d) \!+\! F$ (in dB), where $d$ denotes the two-dimensional distance (in meters) from the user to the \gls{bs}, $f$ is the carrier frequency, and $F$ represents log-normal shadowing with $6$~dB standard deviation; the intra- and inter-cluster shadowing correlations are set to $1$ and $0.1$, respectively. 
The \gls{tdd} frame length is $\tau_{\mathrm{c}}\!=\!200$ channel uses. 
The uplink pilot transmit power budget per user is $P_{\mathrm{ul}}\!=\!20$~dBm, and the \gls{bs} downlink power budget is $P_{\mathrm{dl}}\!=\!33$~dBm. 
The downlink power-control policy is configured to maximize the sum spectral efficiency using \gls{cb} precoding with tilt parameter $\nu\!=\!-0.1$ \cite{2022delaFuente}. 
Spatial covariance matrices are synthesized via a Gaussian local-scattering model: each user's multipath components follow a normal angular distribution centered at the nominal AoA, with $10^{\circ}$ standard deviation. 
Unless otherwise stated, we adopt a noise power spectral density of $-174$~dBm/Hz, a receiver noise figure of $7$~dB, an operating bandwidth of $20$~MHz, and a carrier frequency of $2$~GHz.
Since all users in a given multicast subgroup decode the same stream at the rate dictated by the worst user in that subgroup, the network-wide sum \gls{se} is a convenient metric to highlight the performance impact of subgrouping on multicast delivery. Note that pilot reuse across spatially well-separated subgroups corresponds to choosing a smaller number of multicast subgroups. Lower values of $G$ in our results, therefore, implicitly represent the performance of such pilot-reuse configurations.

We remark that the learning-based selection of $G$ is agnostic to the specific linear precoding rule. Although \gls{cb} is used to generate \gls{se} labels, the pipeline only requires consistent behavior across $G$, not a particular optimality of the precoder. Additional experiments using \gls{zf} precoding confirm that the predicted $G$ values remain stable, showing that the learning model captures structural patterns of the covariance matrices rather than artifacts specific to \gls{cb}.

\begin{table}[!t]
\renewcommand{\arraystretch}{1.2}
\centering
\caption{Simulation setup}
\begin{tabular}{l|c}
\bfseries Parameters & \bfseries Value\\
\hline
 {Number of BS antennas} &  {$64$}\\
 {Maximum DL transmit power} &  {$33$ dBm}\\
 {Maximum UL pilot transmit power} &  {$20$ dBm}\\
 {Path loss model [dB]} &  {$-32.4 - 20\text{log}_{10}(f)$}\\
&  {$- 37.6\log_{10}(d) + F$}\\
 {Angular standard deviation} &  {$10^{\circ}$}\\
 {User noise figure} &  {$7$ dB}\\
 {Noise power spectral density} &  {$-174$ dBm/Hz}\\
 {Cell radius} &  {$200$ m}\\
 {Shadowing standard deviation} &  {$6$ dB}\\
 {Shadowing intra-cluster correlation} &  {$1$}\\
 {Shadowing inter-cluster correlation} &  {$0.1$}\\
 {Carrier frequency} &  {$2$ GHz}\\
 {Operating bandwidth} &  {$20$ MHz}\\
 {Channel coherence samples} &  {$200$}\\
 {Precoding strategy} & CB\\
\hline
\end{tabular}
\label{Tab:sim_param}
\end{table}

\subsection{ASSESSMENT OF THE CDF OF THE SUM SE ACHIEVED BY THE LSTM-ASSISTED MULTICAST SUBGROUPING}

In this subsection, we present the results of the different models and methods for estimating the optimal number of multicast subgroups $G$ to maximize the sum SE $\Psi$. We evaluate the \gls{cdf} of the sum SE of the multicast users. The results were evaluated using a test set consisting of $480$ snapshots that include a generalization of the number of multicast users from $2$ to $120$ and the spatial clusters deployed from $1$ to $120$. 

The test dataset enables a detailed analysis of how well the complete process (estimation of the number of multicast subgroups $G$ and the application of the K-means algorithm) maximizes the sum SE $\Psi$. Specifically, we compare our LSTM-assisted approach with five alternative strategies:
\begin{itemize}
        \item  \textbf{Genie-aided multicast subgrouping}: The number of multicast subgroups $G$ is set equal to the number of spatial user clusters generated in the simulation. After determining this value of $G$, the K-means algorithm is applied to partition the $K$ users into the corresponding $G$ disjoint subgroups. This sub-optimal approach requires that the \gls{bs} exactly knows the users' deployment clusters. This is clearly an ideal assumption. This scheme does not necessarily maximize the sum \gls{se}.  
        \item \textbf{Unicast}: The number of multicast subgroups matches the number of multicast users, $G = K$. That is, one transmission per multicast user. 
        \item \textbf{Conventional multicast}: All users are served through a single multicast transmission, $G = 1$. 
        \item \textbf{Naive multicast subgrouping}: A random value of multicast subgroups $G \!\in\! \mathcal{K}\!=\!\{1, \ldots, K\}$ is selected, followed by the application of K-means.
        \item \textbf{DBSCAN multicast subgrouping}: A density-based clustering algorithm that does not require the number of multicast subgroups $G$ as input, since this method does not utilize K-means. However, DBSCAN is parameter-dependent, allowing for customization of the scenario according to the desired clustering level. Note that these results employ a fixed DBSCAN multicast subgrouping parameterization. DBSCAN is applied to the user-similarity matrix constructed from the metric in (\ref{eq:metric}). Two users are considered neighbors if their similarity exceeds a threshold $\epsilon$, and clusters are formed according to the standard DBSCAN rule with \emph{min\_samples}. In our simulations, we set $\epsilon = 0.55$ and \texttt{min\_samples} = 3, which provided a reasonable balance between cluster compactness and noise-point suppression across the evaluated scenarios. Users labeled as noise points are assigned to the nearest cluster according to the same similarity measure. Each resulting cluster is treated as a multicast subgroup.
\end{itemize}

We also present the \textbf{Upper bound multicast subgrouping} that represents the maximum achievable sum SE within the evaluated range of subgroup configurations under the scalable-CSI model (one pilot per subgroup and multicast CB precoding). This serves as the theoretical upper bound of the sum SE that is also fixed with the input data, and it is achieved through exhaustive search.
Next, we detail the three proposed LSTM-based multicast subgrouping methods:

\begin{itemize}
    \item  \textbf{LSTM (simple model) multicast subgrouping}: The SE predicted by the original neural network model, which estimates the value of \(G\) using a combination of dimensionality reduction, LSTM encoding, and a dense network.
    \item  \textbf{LSTM ($K_\mathrm{max}$) multicast subgrouping}: The SE obtained by the enhanced model replicates the input data to ensure $K_{\mathrm{max}}$ users, regardless of the actual number of users. This modification increases the robustness of the network.
    \item  \textbf{LSTM (TL SE) multicast subgrouping}: The model directly predicts the sum SE as a function of the subgroup configuration, and the optimal number of multicast subgroups is subsequently obtained by selecting the value of $G$ that maximizes the predicted SE. This model uses TL to optimize the sum SE directly, rather than estimating \(G\).
\end{itemize}

For clarity, we explicitly include the performance of conventional multicast ($G=1$) as a baseline in all figures. We split the set of test scenarios into two different subsets. This division allows us to clearly present the improvement over \emph{unicasting}, single-stream multicast, and \emph{DBSCAN} of the proposed learning-based subgrouping strategies.

Recent contributions have proposed low-complexity optimization-based multicast beamforming algorithms (e.g., \cite{2023Zhang}) as well as extremely lightweight heuristics (e.g., \cite{2026Zaher}). These works assume access to per-user instantaneous CSI and operate outside the scalable-CSI framework considered in this paper. Since our focus is on subgroup selection under the one-pilot-per-subgroup constraint, where instantaneous CSI is not available, such methods are not directly comparable in a fair manner. For completeness, we highlight that integrating our subgroup-selection mechanism with these multicast-optimized precoders would require redefining both the CSI acquisition model and the power-allocation strategy. This integration represents a promising research direction but falls beyond the scope of the present work.

\begin{figure}[!t]
\subfloat[Clustered scenario]
{\includegraphics[width=0.5\textwidth]{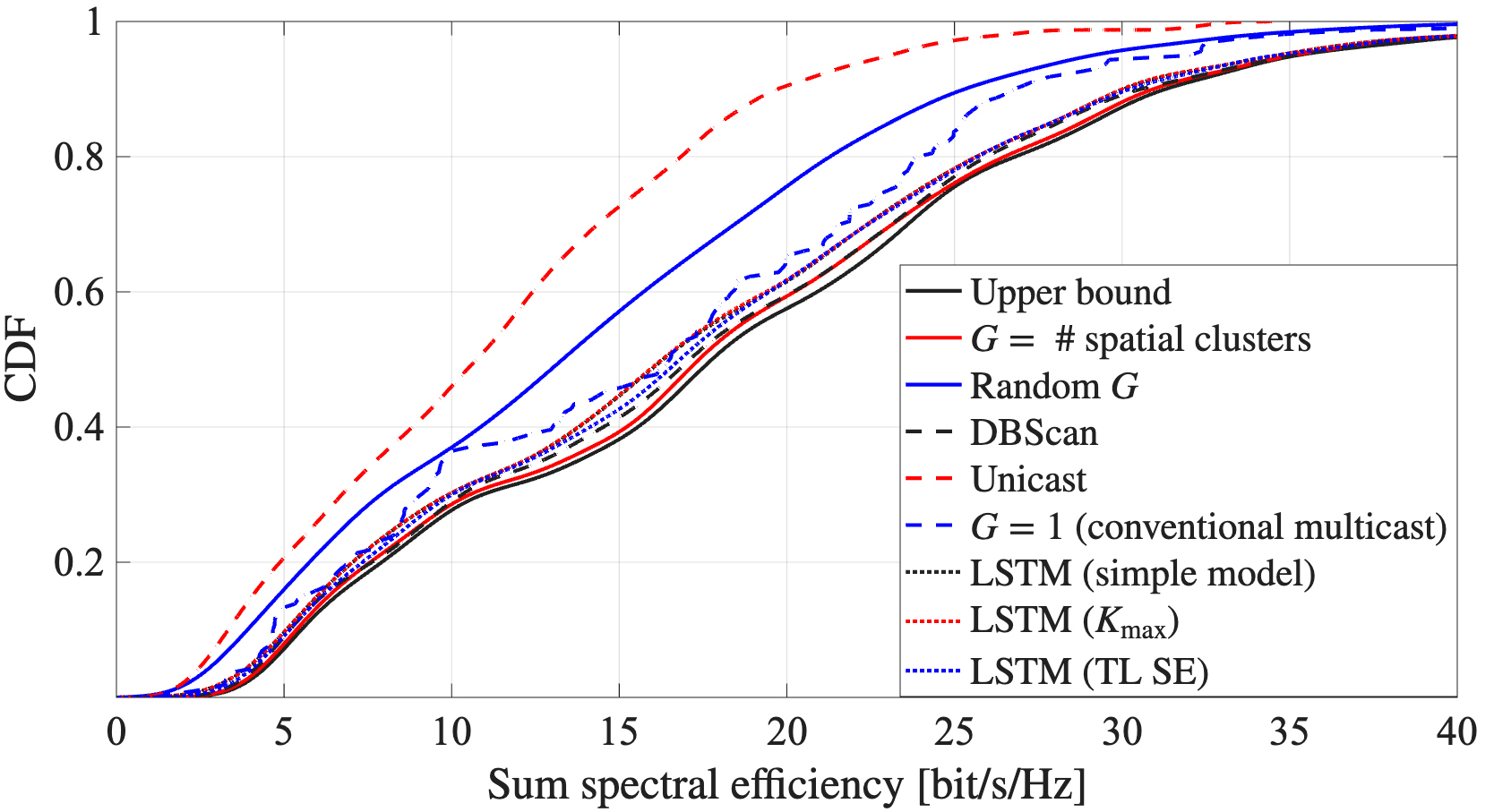}\label{Fig:1-8}}\\
\subfloat[Sparse scenario]
{\includegraphics[width=0.5\textwidth]{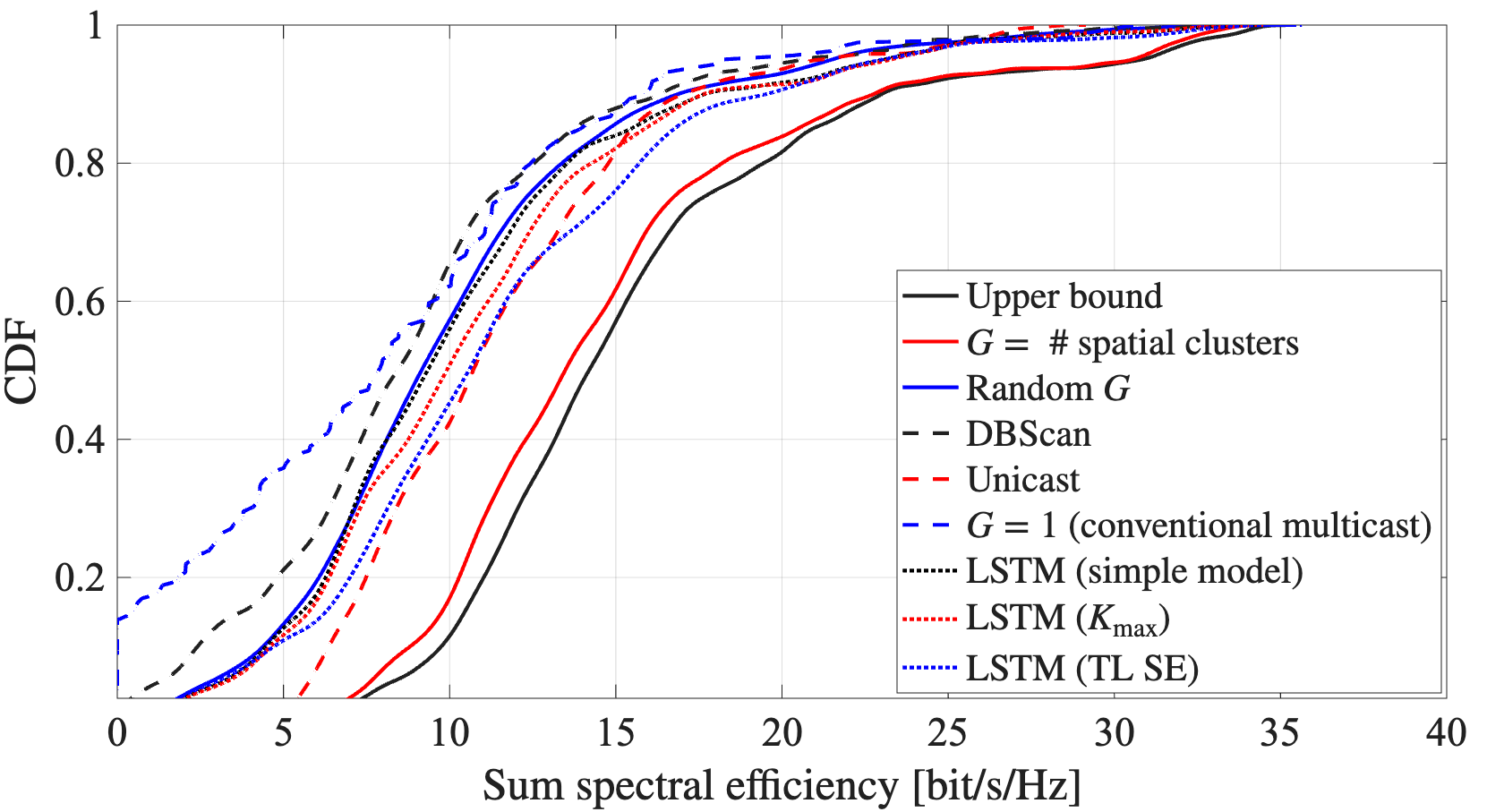}\label{Fig:10-30}}\\
\caption{CDF of the sum \gls{se}. We split the assessed scenarios into a) clustered scenarios where the multicast users are allocated into $\{1, 2, 4, 5\}$ spatial clusters independently of the number of users $K \in \{40, \ldots, 120\}$, and b) sparse scenario where the multicast users are allocated into spatial clusters of $\{1, 2, 4, 5\}$ users independently of the number of users $K \in \{40, \ldots, 120\}$.}
\label{Fig:cdf_percluster}
\end{figure}

Fig. \ref{Fig:1-8} illustrates the CDF of the sum SE when the multicast users $K \in \{40, \ldots, 120\}$ are allocated into a small number of spatial clusters ($\{1, 2, 4, 5\}$), that is, we consider a scenario with a few densely populated areas, referred to as the  \emph{clustered scenario}. Fig. \ref{Fig:1-8} shows that \emph{unicasting} is the worst strategy in this kind of scenario, while \emph{DBSCAN} results in a high performance close to the \emph{upper bound}. The three proposed neural network models present similar results to DBSCAN, enhancing both \emph{random $G$} and \emph{unicasting} techniques. Among the LSTM-based techniques, the \emph{LSTM (TL SE)} results in a slightly higher performance in the clustered scenarios. 

Fig. \ref{Fig:10-30} shows the CDF of the sum SE when the multicast users $K \in \{40, \ldots, 120\}$ are allocated into a large number of spatial clusters, and consequently, allocating a small number of users per cluster ($\{1, 2, 4, 5\}$), that is, we create a scenario with sparsely populated spatial areas which is close to randomly uniform distribution, named \emph{sparse scenario}. Fig. \ref{Fig:10-30} exhibits that the \emph{upper bound} is far from being achieved by the benchmark and the proposed methods. Particularly, \emph{unicasting} presents good performance in this kind of scenario, while the use of \emph{DBSCAN} severely degrades the sum SE performance in sparse scenarios. In sparse scenarios, user covariance matrices exhibit no significant spatial similarity. As a result, forming multicast subgroups yields limited or no gains, and unicast transmission becomes the optimal operating point. This explains why the TL-SE, DL-assisted, and DBSCAN methods converge towards unicast-level performance in Fig. 5b. This behavior is inherent to the nature of multicasting in \gls{mamimo} systems rather than a limitation of the proposed method.
{\color{black}When the spatial distribution of multicast users exhibits limited large-scale correlation, the proposed LSTM-assisted framework naturally predicts a number of multicast subgroups equal to the number of users ($G = K$). In this limiting case, each subgroup contains a single user, and the resulting operation reduces to conventional unicast transmission with per-user pilots and precoders. This behavior emerges from the joint action of the LSTM-based subgroup number estimation and the subsequent user grouping procedure, and is effectively determined by the inferred subgroup cardinality.}

We note that the proposed LSTM-based models, especially the \emph{LSTM (TL SE)}, result in a good trade-off between clustered and sparse scenarios, allowing the system to achieve as good performance as \emph{DBSCAN} in clustered scenarios, and as high performance as \emph{unicasting} in sparse scenarios.

Next, we evaluate the \gls{cdf} of the sum SE achieved employing the complete set of scenarios, varying the number of multicast users $K \in \{40, \ldots, 120\}$ and the number of spatial clusters to distribute multicast users between $\{1, \ldots, 120\}$, that is, from a single spatial cluster to a uniform distribution of users. Fig. \ref{Fig:cdf_general} presents the CDF of the sum SE obtained for the benchmark and proposed models and methods for a general set of scenarios. We observe that using the strategy of creating as many subgroups $G$ as the number of spatial clusters launched in the simulations is close to achieving the upper bound of the sum SE. {\color{black}Although the number of spatial clusters and the SE-optimal number of multicast subgroups are conceptually different, their values tend to align in scenarios with well-defined spatial clustering, which explains the near-optimal performance observed with this strategy.}

The proposed LSTM-based models outperform the sum SE results of the benchmark solutions in the complete range of the CDF. The LSTM (TL SE) method yields better sum SE performance across all percentiles. Static clustering requires scenario-dependent tuning (e.g., density thresholds) and exhibits complementary but limited operating regimes: DBSCAN excels in highly clustered deployments yet degrades in sparse ones, whereas unicast behaves conversely. Trained on diverse covariance patterns and fed by snapshot-specific \gls{pca}, the proposed \gls{lstm} maintains consistent performance across clustered, sparse, and mixed scenarios. The \gls{tl} head further narrows the gap to the upper bound by predicting the \gls{se} profile across $G$ without exhaustive evaluation.

Looking at Fig. \ref{Fig:percentile_general}, we can see these results for the sum SE probabilities of $ 90\%$, $ 80\%$, $ 50\%$, and $30\%$. These results show how, among the benchmark methods, \emph{unicasting} is close to the proposed LSTM-based models' results for the $90\%$-likely sum SE; however, \emph{unicasting} degrades its performance for the best $30\%$-likely sum SE. On the other hand, \emph{DBSCAN} is close to the proposed LSTM-based models' results for the best $30\%$-likely sum SE; however, \emph{DBSCAN} degrades its performance for the $90\%$-likely sum SE. Finally, using a \emph{random $G$} presents poorer sum SE results than the proposed LSTM-based methods in all percentiles.   

\begin{figure}[!t]
\centering
\includegraphics[width=\linewidth]{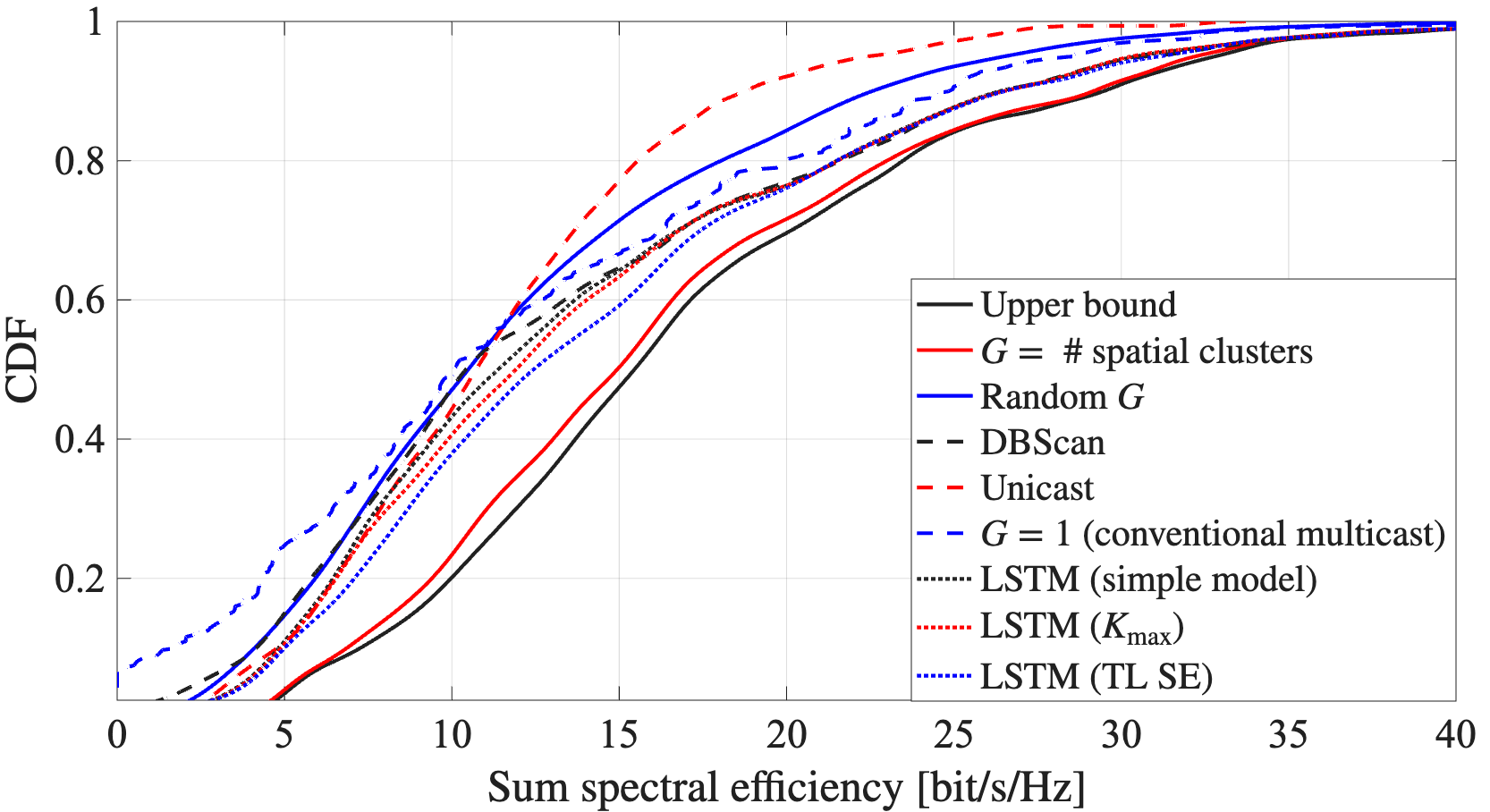}
\caption{CDF of the sum \gls{se} achieved by the assessed models and methods to select the number of multicast subgroups. Evaluation results for generalized scenarios where $K \in \{40, \ldots, 120\}$ and the number of spatial clusters of multicast users is between $\{1, \ldots, 120\}$.}
\label{Fig:cdf_general}
\end{figure}

\begin{figure}[!t]
\centering
\includegraphics[width=\linewidth]{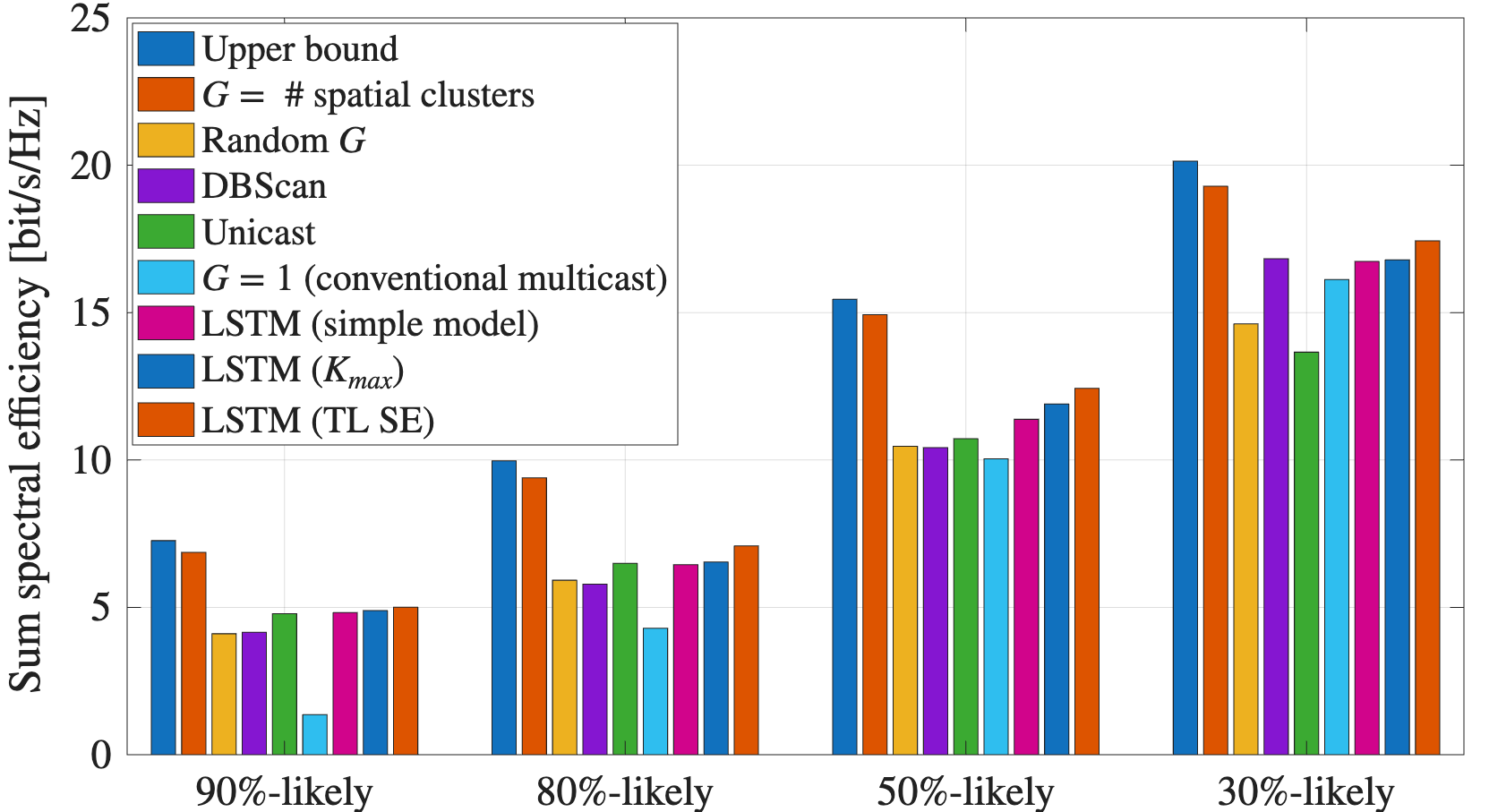}
\caption{$90\%$-likely, $80\%$-likely, $50\%$-likely, and $30\%$-likely sum \gls{se} achieved by the assessed models and methods to select the number of multicast subgroups. Evaluation results for generalized scenarios where $K \in \{40, \ldots, 120\}$ and the number of spatial clusters of multicast users is between $\{1, \ldots, 120\}$.}
\label{Fig:percentile_general}
\end{figure}

In conclusion, the proposed \emph{LSTM (TL SE)} provides a robust and near-optimal solution, independently of the multicast users' distribution, to estimate the number of subgroups $G$ to apply in the K-means algorithm to optimize the sum SE of the multicast service.

\subsection{STATISTICS OF THE SUM SE IN LSTM-ASSISTED MULTICAST SUBGROUPING}

The proposed LSTM-based multicast subgrouping models exhibit a different source of variability. Since the LSTM processes the user descriptors sequentially, its hidden and memory states may depend on the order in which the users' covariance information is introduced. Each row of the PC-score matrix represents one multicast user; therefore, changing the order of these rows changes the sequence processed by the recurrent network, even though the covariance matrices, the PC-score information, and the underlying spatial user distribution remain unchanged. Different permutations of the same user set may consequently lead to slightly different latent representations and estimates of the number of multicast subgroups.

To reduce this order dependence, the users within every training snapshot are randomly shuffled during each training epoch, as described in Section~\ref{subsec:training_process}. This exposes the network to different permutations of equivalent user sets and prevents it from relying on a fixed input ordering. Nevertheless, since the LSTM encoder is not permutation invariant by construction, some residual sensitivity to the user processing order may remain at inference time.

To quantify this effect, each LSTM-based model was evaluated $6000$ times using randomized user orders. In these evaluations, the user covariance information was preserved, and only the order in which the corresponding PC-score vectors were supplied to the LSTM was changed. This procedure allowed us to calculate the average, maximum, minimum, variance, and standard deviation of the resulting sum \gls{se} $\Psi$. The same number of evaluations was used for the \emph{Naive multicast subgrouping} method, although in that case the variability is caused by the random selection of $G$ rather than by the order of the input sequence. In contrast, \emph{Unicast}, \emph{Conventional multicast}, \emph{DBSCAN multicast subgrouping}, \emph{Genie-aided multicast subgrouping}, and \emph{Upper bound multicast subgrouping} produce a fixed sum-\gls{se} value for the considered evaluation dataset and system configuration.

{\color{black}
To complement the spectral-efficiency analysis, the deviation between the predicted number of multicast subgroups and the \gls{se}-optimal subgroup configuration is evaluated separately.
}

Table~\ref{table:sum_SE} summarizes the average sum \gls{se} and the corresponding statistical measures obtained for the evaluated models and methods.

\begin{table*}[t!]
\centering
\caption{Average sum SE and statistical measures for different models and methods.}
\resizebox{\textwidth}{!}{%
\begin{tabular}{|l|c|c|c|c|c|c|}
\hline
\textbf{Model/Method} & \textbf{Avg. $\Psi$ [bit/s/Hz]} & \textbf{\% of Upper bound} & \textbf{Variance} & \textbf{Sigma} & \textbf{Min $\Psi$ [bit/s/Hz]} & \textbf{Max $\Psi$ [bit/s/Hz]} \\
\hline
Unicast ($G = K$)& 11.3868 & 67.2 & 0 & 0 & 11.3868 & 11.3868 \\
Naive multicast subgrouping & 12.0397 & 71.1 & 0.011 & 0.1047 & 11.8185 & 12.4264 \\
{Conventional multicast ($G{=}1$)} & 12.0939 & 71.4 & 0 & 0 & 12.0939 & 12.0939 \\
DBSCAN multicast subgrouping & 13.4204 & 79.2 & 0 & 0 & 13.4204 & 13.4204 \\
LSTM (simple model) & 13.7961 & 81.4 & 0.0013 & 0.0358 & 13.6725 & 13.9528 \\
LSTM ($K_\mathrm{max}$) & 13.9747 & 82.5 & 0.0013 & 0.0354 & 13.8630 & 14.1181 \\
LSTM (TL SE) & 14.4383 & 85.2 & 0.0056 & 0.0749 & 14.1518 & 14.7759 \\
Genie-aided multicast subgrouping & 16.3907 & 96.8 & 0 & 0 & 16.3907 & 16.3907 \\
Upper bound multicast subgrouping & 16.9385 & 100 & 0 & 0 & 16.9385 & 16.9385 \\
\hline
\end{tabular}
}
\label{table:sum_SE}
\end{table*}

As shown in Table~\ref{table:sum_SE}, the maximum achievable sum \gls{se} within the evaluated range of multicast subgroup configurations is $16.9385$~bit/s/Hz. This value is obtained by selecting, for each evaluated scenario, the value of $G$ that maximizes the sum \gls{se} under the considered scalable-\gls{csi} model, based on one pilot per subgroup and multicast \gls{cb} precoding. Therefore, it constitutes an upper bound for the subgroup-selection problem analyzed in this work. It should not be interpreted as the absolute theoretical upper bound of multicast beamforming, which would require per-user instantaneous \gls{csi} and fully optimized multicast precoding.

The proposed and benchmark methods are evaluated according to how closely they approach this reference value. The \emph{Genie-aided multicast subgrouping} method uses a priori knowledge of the true number of spatial user clusters and applies K-means using this value. It achieves a sum \gls{se} of $16.3907$~bit/s/Hz, corresponding to $96.8\%$ of the evaluated upper bound. However, the true number of spatial clusters is known only as part of the simulated data-generation process and is not available in a practical system. Thus, the genie-aided result provides a reference for assessing the ability of the proposed models to infer a suitable number of multicast subgroups from the users' spatial covariance information.

Among the methods that do not rely on such a priori information, the LSTM-based approaches consistently outperform the benchmark strategies. The LSTM simple model, the LSTM ($K_{\mathrm{max}}$) model, and the LSTM (TL-SE) model achieve average sum-\gls{se} values of $13.7961$, $13.9747$, and $14.4383$~bit/s/Hz, respectively, whereas the best-performing practical benchmark, DBSCAN multicast subgrouping, achieves $13.4204$~bit/s/Hz. The LSTM (TL-SE) model provides the best average result among the practical methods, reaching $85.2\%$ of the evaluated upper bound.

The statistical results also show that the impact of changing the user processing order is limited. The standard deviations obtained by the LSTM simple model, the LSTM ($K_{\mathrm{max}}$) model, and the LSTM (TL-SE) model are $0.0358$, $0.0354$, and $0.0749$~bit/s/Hz, respectively. These values correspond to approximately $0.26\%$, $0.25\%$, and $0.52\%$ of their respective average sum-\gls{se} values.

More importantly, the residual order dependence does not alter the main comparative conclusions. Even the minimum sum-\gls{se} values observed over the $6000$ randomized user-order evaluations are $13.6725$, $13.8630$, and $14.1518$~bit/s/Hz for the LSTM simple model, the LSTM ($K_{\mathrm{max}}$) model, and the LSTM (TL-SE) model, respectively. All these minimum values remain above the $13.4204$~bit/s/Hz achieved by DBSCAN multicast subgrouping. Therefore, the proposed models outperform the main practical benchmark even under the least favorable user ordering observed in the evaluation. These results show that the randomization strategy used during training, together with the learned covariance-based representation, effectively limits the sensitivity of the recurrent models to the input sequence.

{\color{black}
\begin{table*}[!t]
\centering
\caption{{\color{black}Average number of predicted multicast subgroups and deviation from the optimal subgroup configuration.}}
\footnotesize
\begin{tabular}{|l|c|c|c|c|}
\hline
{\color{black}\textbf{Model/Method}} & {\color{black}\textbf{Avg. $G$}} & {\color{black}\textbf{Avg. $|G - G^\star|$}} & {\color{black}\textbf{Avg. $\frac{|G - G^\star|}{K}$}} & {\color{black}\textbf{Std. $(G - G^\star)$}}\\
\hline
{\color{black}Unicast ($G = K$)} & {\color{black} 80.0 }&{\color{black} 67.9 }&{\color{black} 0.82 }&{\color{black} 34.6 }\\
{\color{black}Naive multicast subgrouping }&{\color{black} 40.0 }&{\color{black} 33.8 }&{\color{black} 0.43 }&{\color{black} 33.2 }\\
{\color{black}Conventional multicast ($G = 1$)} &{\color{black} 1.0 }&{\color{black} 11.1 }&{\color{black} 0.16 }&{\color{black} 18.4 }\\
{\color{black}DBSCAN multicast subgrouping }&{\color{black} 4.3 }&{\color{black} 9.2 }&{\color{black} 0.13 }&{\color{black} 17.3 }\\
{\color{black}LSTM (simple model) }&{\color{black} 11.6 }&{\color{black} 9.5 }&{\color{black} 0.12 }&{\color{black} 18.5 }\\
{\color{black}LSTM ($K_\mathrm{max}$) }&{\color{black} 12.3 }&{\color{black} 9.8 }&{\color{black} 0.13 }&{\color{black} 19.7 }\\
{\color{black}LSTM (TL-SE) }&{\color{black} 6.7 }&{\color{black} 10.5 }&{\color{black} 0.14 }&{\color{black} 20.5 }\\
{\color{black}Genie-aided multicast subgrouping }&{\color{black} 21.0 }&{\color{black} 9.6 }&{\color{black} 0.10 }&{\color{black} 23.3 }\\
{\color{black}Upper bound multicast subgrouping }&{\color{black} 12.1 }&{\color{black} $0$ }&{\color{black} $0$ }&{\color{black} $0$ }\\
\hline
\end{tabular}
\label{table:G_analysis}
\end{table*}
}

\begin{table*}[!t]
\centering
\caption{{\color{black}Performance of multicast subgrouping methods in CF-mMIMO scenarios for different pilot lengths ($\tau_\mathrm{p}$).}}
\footnotesize
\setlength{\tabcolsep}{4pt}
%\renewcommand{\arraystretch}{1}
%\resizebox{\textwidth}{!}{
\begin{tabular}{|l|ccc|ccc|}
\hline
\multirow{2}{*}{{\color{black}\textbf{Model/Method}}} 
& \multicolumn{3}{c|}{{\color{black}\textbf{$\tau_\mathrm{p} = 10$}}} 
& \multicolumn{3}{c|}{{\color{black}\textbf{$\tau_\mathrm{p} = 20$}}} \\
\cline{2-7}

& {\color{black}\textbf{Avg. $\Psi$}} 
& {\color{black}\textbf{Avg. $G$}} 
& {\color{black}\textbf{Avg. $|G - G^\star|$}} 
& {\color{black}\textbf{Avg. $\Psi$}} 
& {\color{black}\textbf{Avg. $G$}} 
& {\color{black}\textbf{Avg. $|G - G^\star|$}} \\
\hline

{\color{black}Unicast ($G = K$)}          
& {\color{black}72.5} & {\color{black} 80.0} & {\color{black}46.0} 
& {\color{black}84.26} & {\color{black} 80.0} & {\color{black}38.2} \\

{\color{black}Naive multicast subgrouping} 
& {\color{black}71.6} & {\color{black}40.0} & {\color{black}37.5} 
& {\color{black}80.11} & {\color{black}40.0} & {\color{black}30.0} \\

{\color{black}Conventional multicast ($G = 1$)} 
& {\color{black}48.25} & {\color{black}1.0} & {\color{black}32.9} 
& {\color{black}43.87} & {\color{black}1.0} & {\color{black}46.2} \\

{\color{black}DBSCAN multicast subgrouping}       
& {\color{black}73.15} & {\color{black}25.3} & {\color{black}20.6} 
& {\color{black}77.20} & {\color{black}24.1} & {\color{black}27.7} \\

{\color{black}LSTM (simple model)}        
& {\color{black}92.6} & {\color{black}27.0} & {\color{black}10.3} 
& {\color{black}90.73} & {\color{black}24.9} & {\color{black}22.5} \\

{\color{black}LSTM ($K_\mathrm{max}$)}     
& {\color{black}93.95} & {\color{black}30.1} & {\color{black}10.6} 
& {\color{black}90.43} & {\color{black}28.2} & {\color{black}20.3} \\

{\color{black}Genie-aided multicast subgrouping} 
& {\color{black}94.71} & {\color{black}20.9} & {\color{black}13.6} 
& {\color{black}93.45} & {\color{black}21.0} & {\color{black}26.6} \\

{\color{black}Upper bound multicast subgrouping} 
& {\color{black}117.12} & {\color{black} 34.0} & {\color{black}$0$} 
& {\color{black}118.25} & {\color{black} 47.1} & {\color{black}$0$} \\

\hline
\end{tabular}
%}
\label{table:CF_results}
\end{table*}

First, we observe the results of the other benchmark methods. The sum SE using \emph{unicast} transmissions is $11.3868$ bit/s/Hz, which is the poorest result when the users' distribution implies a great variety of scenarios, including spatial clusters of users (from $1$ to $120$). 
The method that selects a \emph{random \(G\)} achieves an average sum SE of $12.0397$ bit/s/Hz, which is slightly better than using \emph{unicast} transmissions, but still far below the maximum possible sum SE.
The \emph{DBSCAN} method, which does not require an explicit estimation of \(G\), achieves a sum SE of $13.4204$ bit/s/Hz. This method outperforms both the \emph{unicast} and \emph{random \(G\)} selections, but is still far from reaching the maximum possible sum SE.

We analyze the results of the proposed models based on the neural network. The original model, that is, the \emph{LSTM (simple model)}, achieves an average sum SE of $13.7961$ bit/s/Hz, outperforming \emph{DBSCAN} results. This model also demonstrates low variance ($0.0013$) and a standard deviation of $0.0358$, indicating stable predictions over the $6000$ evaluations.
The model with replicated inputs, that is, the \emph{LSTM ($K_\mathrm{max}$)}, improves the average sum SE to $13.9747$ bit/s/Hz. This improvement suggests that input replication enhances the model's robustness when processing data with varying user counts. The model also maintains low variance and standard deviation ($0.0354$), which confirms the stability of the performance.
Finally, the model trained specifically to estimate the sum SE, that is, the \emph{LSTM (TL SE)}, demonstrates the best performance among neural networks, with an average sum SE of $14.4383$ bit/s/Hz. Although the variance ($0.0056$) and the standard deviation ($0.0749$) increase, this model provides the closest approximation to the maximum possible sum SE. This result highlights the effectiveness of training the model to optimize the sum SE directly rather than solely estimating the number of clusters.

{\color{black}
In addition to the \gls{se} results, Table~\ref{table:G_analysis} reports the average number of predicted subgroups and their deviation from the \gls{se}-optimal configuration. The standard deviation of the absolute error is also provided to assess the robustness of each method across different spatial realizations. Here, $G^\star$ denotes the number of multicast subgroups that maximizes the sum \gls{se} (upper bound multicast subgrouping).
The results show that the LSTM-based methods achieve consistently low absolute and normalized errors with respect to the optimal subgroup configuration, indicating that the learned model effectively captures the underlying spatial structure of the users. In contrast, baseline methods such as random grouping exhibit significantly larger deviations. Furthermore, the low standard deviation of the error confirms the stability of the proposed approach across different scenarios, demonstrating its robustness to variations in user distributions.
}

{\color{black}These results indicate that the predicted number of subgroups closely follows the trend of the SE-optimal configuration across different scenarios, even if exact matching is not required to achieve near-optimal performance.}

In conclusion, the LSTM-based models, particularly those incorporating data replication and direct sum SE estimation, outperform traditional methods such as \emph{unicasting}, \emph{DBSCAN}, and \emph{random \(G\)} selection. The improvement in the average sum SE indicates that these networks are better equipped to capture the clustering structure, approaching the theoretical upper bound of the sum spectral efficiency.

{\color{black}
\subsection{ROBUSTNESS UNDER PERTURBED COVARIANCE INFORMATION}

In practical deployments, the spatial covariance matrices $\boldsymbol{R}_k$ are not perfectly known, since they must be estimated from received signals over finite observation intervals. As a result, the covariance information available at the base station may be affected by estimation errors, temporal mismatch, and changes in the propagation environment. Since the proposed subgrouping framework relies on the spatial covariance matrices as input information, it is important to evaluate its robustness under imperfect covariance knowledge.

To this end, we introduce controlled perturbations in the covariance matrices used as input to the subgrouping algorithm. The perturbed covariance matrix associated with user $k$ is modeled as
\begin{equation}
\tilde{\boldsymbol{R}}_k =
(1-\alpha)\boldsymbol{R}_k + \alpha \boldsymbol{E}_k,
\end{equation}
where $\boldsymbol{R}_k$ denotes the true covariance matrix and $\boldsymbol{E}_k$ is a random positive semi-definite matrix. The perturbation matrix is generated as
\begin{equation}
\boldsymbol{E}_k =
\frac{1}{M}\boldsymbol{Z}\boldsymbol{Z}^{\herm},
\end{equation}
where $\boldsymbol{Z}$ has independent and identically distributed complex Gaussian entries and $M$ denotes the number of antennas. After generation, $\boldsymbol{E}_k$ is normalized to have the same trace as $\boldsymbol{R}_k$, i.e.,
\begin{equation}
\boldsymbol{E}_k \leftarrow
\frac{\mathrm{tr}(\boldsymbol{R}_k)}
{\mathrm{tr}(\boldsymbol{E}_k)}
\boldsymbol{E}_k.
\end{equation}
The parameter $\alpha \in [0,1]$ controls the level of mismatch between the true covariance matrix and the covariance matrix used by the learning-based subgrouping method. Therefore, $\alpha=0$ corresponds to the nominal case with perfect covariance information, whereas larger values of $\alpha$ represent increasing levels of covariance uncertainty.

This perturbation model is specifically designed to preserve the fundamental structure of covariance matrices. Since both $\boldsymbol{R}_k$ and $\boldsymbol{E}_k$ are positive semi-definite, their convex combination $\tilde{\boldsymbol{R}}_k$ is also positive semi-definite. Hence, the perturbed matrices remain valid covariance matrices. This is in contrast to a direct additive Gaussian perturbation, which could destroy the positive semi-definite structure and produce physically inconsistent matrices. The proposed model can therefore be interpreted as a controlled mismatch between the true spatial covariance matrix and its estimate, while preserving the mathematical properties required for covariance-based signal processing.

In the robustness evaluation, the subgrouping decision is obtained from the perturbed covariance matrices $\tilde{\boldsymbol{R}}_k$, whereas the resulting sum spectral efficiency is evaluated using the true covariance matrices $\boldsymbol{R}_k$. This procedure allows us to isolate the impact of imperfect covariance information on the subgrouping decision itself.

Table~\ref{tab:robustness_covariance_errors} reports the average sum spectral efficiency obtained by the three proposed LSTM-based configurations for different values of $\alpha$. The table also includes the percentage with respect to the upper bound reported in Table~\ref{table:sum_SE}. The case $\alpha=0$ corresponds to the nominal performance without covariance perturbations.

\begin{table*}[t!]
\centering
\caption{Robustness of the LSTM-based models under perturbed covariance information.}
\label{tab:robustness_covariance_errors}
\resizebox{\textwidth}{!}{%
\begin{tabular}{|c|ccc|ccc|}
\hline
\multirow{2}{*}{$\boldsymbol{\alpha}$}
&
\multicolumn{3}{c|}{\textbf{Avg. $\boldsymbol{\Psi}$ [bit/s/Hz]}}
&
\multicolumn{3}{c|}{\textbf{\% of Upper bound}}
\\
\cline{2-7}
&
\textbf{LSTM simple model}
&
\textbf{LSTM ($K_{\mathrm{max}}$)}
&
\textbf{LSTM (TL SE)}
&
\textbf{LSTM simple model}
&
\textbf{LSTM ($K_{\mathrm{max}}$)}
&
\textbf{LSTM (TL SE)}
\\
\hline
0   & 13.7961 & 13.9747 & 14.4383 & 81.4 & 82.5 & 85.2 \\
0.1 & 13.4179 & 13.6384 & 14.0781 & 79.2 & 80.5 & 83.1 \\
0.2 & 13.1042 & 13.3616 & 14.0389 & 77.4 & 78.9 & 82.9 \\
0.3 & 12.9593 & 13.2524 & 13.4522 & 76.5 & 78.2 & 79.4 \\
0.4 & 12.8155 & 13.0246 & 13.7209 & 75.7 & 76.9 & 81.0 \\
0.5 & 12.8177 & 13.0047 & 13.5727 & 75.7 & 76.8 & 80.1 \\
\hline
\end{tabular}
}
\end{table*}

The results show that the three LSTM-based configurations remain reasonably stable as the covariance mismatch increases. As expected, the use of perturbed covariance matrices reduces the achievable sum spectral efficiency with respect to the nominal case $\alpha=0$, since the subgrouping decision is progressively based on spatial statistics that deviate from the true covariance information. However, the degradation is gradual and moderate for all the evaluated models.

For the simple LSTM model, the average sum spectral efficiency decreases from $13.7961$ bit/s/Hz in the nominal case to $12.8177$ bit/s/Hz for $\alpha=0.5$, which corresponds to a relative degradation of approximately $7.09\%$. For the fixed-size LSTM model with $K_{\mathrm{max}}$, the performance decreases from $13.9747$ bit/s/Hz to $13.0047$ bit/s/Hz, corresponding to a relative degradation of approximately $6.94\%$. Finally, the transfer-learning model decreases from $14.4383$ bit/s/Hz to $13.5727$ bit/s/Hz, which represents a relative degradation of approximately $5.99\%$. Therefore, the LSTM model adapted through transfer learning exhibits the smallest relative degradation under the strongest perturbation level considered.

The fixed-size LSTM model with $K_{\mathrm{max}}$ consistently outperforms the simple LSTM model for all values of $\alpha$. This indicates that the fixed-size input representation provides a more robust structure for the recurrent encoder, even though it requires replication when the number of users is lower than $K_{\mathrm{max}}$. The replicated representation appears to help the LSTM extract more stable structural information from the covariance data, which results in improved robustness under imperfect covariance knowledge.

The transfer-learning model achieves the best performance across all evaluated perturbation levels. This behavior is consistent with its training strategy, since the dense output stage is adapted to select the subgrouping configuration according to the resulting sum spectral efficiency, rather than only according to the estimated number of clusters. As a consequence, the final decision is more directly aligned with the communication performance metric. Although the performance is not strictly monotonic with respect to $\alpha$, which can occur due to the random nature of the perturbations and the discrete nature of the subgrouping decision, the transfer-learning model maintains a clear advantage over the other LSTM-based configurations throughout the evaluated range.

It is also relevant to compare these perturbed results with the baseline methods reported in Table~\ref{table:sum_SE}. The simple LSTM model remains above the unicast configuration, the naive multicast subgrouping strategy, and the conventional multicast scheme for all the considered perturbation levels. However, for large covariance mismatch values, its performance becomes slightly lower than the DBSCAN-based multicast subgrouping baseline. In contrast, the fixed-size LSTM model remains close to the DBSCAN baseline even under strong perturbations, while clearly outperforming the unicast, naive multicast, and conventional multicast schemes.

The transfer-learning model provides the strongest robustness margin. Even for the largest perturbation level, $\alpha=0.5$, it achieves an average sum spectral efficiency of $13.5727$ bit/s/Hz, which remains above the DBSCAN baseline value of $13.4204$ bit/s/Hz. Moreover, at $\alpha=0.5$, the transfer-learning model still preserves approximately $80.1\%$ of the upper-bound performance, which is slightly higher than the $79.2\%$ achieved by DBSCAN in the nominal baseline evaluation. This shows that the proposed transfer-learning strategy remains competitive with the strongest non-learning baseline even when the covariance matrices used for subgrouping are significantly perturbed.

Overall, the robustness analysis confirms that the proposed LSTM-based subgrouping framework does not rely on perfectly estimated covariance matrices. Although imperfect covariance information naturally reduces the achievable sum spectral efficiency, the degradation is progressive and the proposed models remain competitive under substantial covariance mismatch. In particular, the fixed-size LSTM representation improves robustness with respect to the simple LSTM model, while the transfer-learning adaptation provides the best overall trade-off between nominal performance and resilience to covariance estimation errors.
}

{\color{black}
\subsection {GENERALIZATION TO CF-mMIMO SCENARIOS}
To further assess the generalization capability of the proposed learning-based multicast subgrouping approach, we extend the evaluation to a \gls{cf-mmimo} setting. The objective of this analysis is not to redesign or optimize the system under this architecture, but rather to evaluate whether the subgroup structure learned in the single-cell scenario can be effectively transferred to a distributed deployment.

In this case, we consider a CF-mMIMO network composed of $L=100$ distributed access points (APs), each equipped with $N=4$ antennas, serving a varying number of single-antenna users with $K=\{40, \ldots, 120\}$. APs and users are deployed over a square area of side $1000$~m with wrap-around to avoid boundary effects. 

Two pilot lengths are considered, namely $\tau_\mathrm{p} = 10$ and $\tau_\mathrm{p} = 20$, in order to evaluate different levels of pilot contamination, while maintaining consistency with the spatial user distributions used in the single-cell scenarios.

Although the CF-mMIMO channel model does not exhibit the same structured spatial correlation as in the single-cell case, the spatial distribution of users across the network still induces implicit large-scale statistical patterns that can be exploited for subgroup inference from second-order statistics.

The LSTM-based models, trained exclusively on single-cell massive MIMO data, are directly applied to these CF-mMIMO scenarios without retraining. Specifically, we evaluate the LSTM (simple model) and the LSTM ($K_\mathrm{max}$), and compare them against conventional baseline strategies, including unicast transmission, conventional multicast, DBSCAN-based subgrouping, and naive multicast subgrouping with random selection of $G$. The transfer learning model aimed at spectral efficiency estimation is intentionally excluded, since the mapping between subgroup configuration and performance depends on system-specific aspects such as interference management, precoding design, and power allocation.

Table~\ref{table:CF_results} summarizes the obtained results in terms of average sum spectral efficiency, the average number of predicted subgroups, and their deviation from the SE-optimal configuration, defined by the subgroup number $G^\star$ that maximizes the sum spectral efficiency.

The results show that the proposed LSTM-based models maintain a small deviation from the optimal subgroup configuration across all considered CF-mMIMO scenarios and pilot lengths, even without retraining. This confirms that the learned representation effectively captures the underlying spatial structure of the users, which is preserved across different network architectures. As a result, the inferred subgroup configuration remains close to the optimal one, leading to competitive spectral efficiency performance.

In contrast, baseline methods exhibit larger deviations from the optimal subgroup configuration, resulting in a less efficient exploitation of the spatial structure and lower spectral efficiency. Overall, these results support the claim that the proposed learning-based subgrouping strategy generalizes beyond the single-cell regime at the level of structural inference, while maintaining robustness under different system configurations.
}

\subsection{COMPUTATIONAL COMPLEXITY ANALYSIS}

The proposed LSTM-assisted multicast subgrouping framework involves three main computational blocks: PCA-based preprocessing, neural-network inference, and clustering-based subgroup assignment. Below, we provide a concise complexity analysis for each component.
\subsubsection{PCA PREPROCESSING}
For each snapshot, PCA is applied to the $K$ spatial covariance matrices $\{R_k\}$. 
Computing the principal components requires performing an eigenvalue decomposition of an $M \times M$ matrix, leading to a per-snapshot complexity of
\begin{equation}
    O(K M^2).
\end{equation}

\subsubsection{LSTM INFERENCE}
The LSTM layer processes $K$ feature vectors of length $K_{\max}$. With $H = K_{\max}$ hidden units, the forward-pass complexity scales as
\begin{equation}
    O(K K_{\max}^2),
\end{equation}
which is lightweight compared to the PCA step.
\subsubsection{DENSE NEURAL NETWORK INFERENCE}
The dense layers contain at most $O(K_{\max}^2)$ parameters. Their computational cost is negligible relative to the PCA and LSTM steps.
\subsubsection{SUBGROUPING REFINEMENT VIA K-MEANS}
The K-means refinement used to assign users to the predicted number of multicast subgroups incurs a complexity of
\begin{equation}
    O(K G T M^2),
\end{equation}
where $G$ is the number of subgroups and $T$ is the number of clustering iterations (typically $T < 10$).
\subsubsection{DBSCAN BASELINE}
The density-based DBSCAN clustering used as a benchmark operates on the pairwise similarity structure derived from the metric in (\ref{eq:metric}). This results in a worst-case complexity of
\begin{equation}
    O(K^2 M^2).
\end{equation}

Overall, the proposed method scales polynomially in both $K$ and $M$, and remains significantly less complex than optimization-based multicast beamforming approaches (e.g., SDR- or SCA-based designs). 
Low-complexity multicast beamforming techniques proposed in  \cite{2023Zhang,2026Zaher} significantly reduce the computational burden of classical SDR/SCA solvers. However, all these methods share a fundamental prerequisite: full per-user instantaneous CSI must be available at the BS at each coherence block and typically incur cubic or higher-order computational cost in $K$ and $M$. This requirement implies one mutually orthogonal UL pilot per user, resulting in UL training overheads that scale linearly with $K$, which is incompatible with the scalable CSI-acquisition model assumed in this work. The proposed subgrouping framework is therefore well suited for large-scale \gls{mamimo} systems operating under practical \gls{csi} acquisition constraints.

For clarity, we summarize the key differences:
\begin{itemize}
    \item \textbf{Proposed framework:} requires one pilot per subgroup, operates with large-scale spatial covariances only, and scales polynomially with $K$ and $M$; thus suitable for massive-antenna multicasting with short coherence intervals.
    \item \textbf{Low-complexity multicast beamformers \cite{2023Zhang,2026Zaher}:} 
    require instantaneous CSI per user, incur frequent re-optimization when channels vary, and cannot operate under one-pilot-per-subgroup constraints.
\end{itemize}

Therefore, although these multicast-optimized precoders constitute valuable baselines in instantaneous-CSI settings, they are not directly applicable to the scalable-CSI regime addressed in this paper. Our results should thus be interpreted within the multicasting scenario where instantaneous CSI is fundamentally unavailable, making subgroup selection the primary performance-limiting factor.

\section{CONCLUSION AND FUTURE PERSPECTIVES}
\label{sec:conclusion}

This work has presented a novel LSTM-assisted multicast subgrouping strategy for massive MIMO systems, leveraging snapshot-specific PCA and deep learning to estimate the optimal number of subgroups that maximize the overall sum spectral efficiency. The proposed mechanism preserves multicast semantics, one content, one service, while improving physical-layer efficiency through multiple precoders adapted to users' channel-statistics similarity. To the best of our knowledge, this is the first framework that jointly combines \emph{(i)} snapshot-specific dimensionality reduction, \emph{(ii)} sequential modeling of users' covariance structure, and \emph{(iii)} transfer learning for SE prediction in multicast subgrouping, thereby eliminating exhaustive search at inference and enabling strong generalization.

The presented framework is designed for scenarios where full per-user \gls{csi} and large-scale optimization-based multicast precoding are not implementable and, consequently, the results should be interpreted within the pilot-limited \gls{mamimo} regime.

The developed architecture effectively maps spatial covariance information into a fixed-dimensional latent representation, enabling robust estimation of the multicast structure even under variable user distributions and cluster formations. In particular, the proposed approach correctly identifies when subgrouping is beneficial (clustered scenarios) and when it is not (sparse scenarios), adapting the number of subgroups to the spatial structure of the environment. The base LSTM classifier and its \gls{tl} extension avoid exhaustive search over all subgroup configurations, substantially reducing computational complexity. Moreover, the \gls{tl} approach significantly lowers the amount of labeled training data required to estimate spectral efficiency by reusing a pre-trained encoder and fine-tuning only a lightweight output head.

Extensive simulations across a wide range of user distributions demonstrate the effectiveness of the proposed approach. The base LSTM classifier, trained solely on synthetic cluster labels, already surpasses conventional multicasting, unicast transmission, random subgrouping, and density-based clustering (DBSCAN), achieving over 81\% of the maximum achievable sum spectral efficiency. A replication-based enhancement further increases robustness to variable input lengths, raising the performance to 82.5\%. Finally, the \gls{tl} extension, which directly optimizes the SE prediction, reaches up to 85.2\% of the theoretical optimum without requiring exhaustive subgroup search during inference.

A key outcome of this work is that the proposed model learns the underlying spatial structure governing multicast subgrouping. This is further supported by the additional validation in CF-mMIMO scenarios, where the models—trained exclusively on single-cell data—maintain a small deviation from the SE-optimal subgroup configuration. This result confirms that subgroup-number inference generalizes across system architectures, while the mapping from subgroup configuration to spectral efficiency remains system-dependent.

From a system perspective, the proposed method provides a robust and adaptive solution across heterogeneous deployment scenarios. While unicast transmission performs well in sparse user distributions and DBSCAN clustering excels in highly clustered scenarios, both degrade significantly outside their optimal conditions. In contrast, the LSTM-based subgrouping strategy consistently delivers high spectral efficiency across diverse spatial configurations, making it particularly suitable for realistic deployments with dynamic and unpredictable user topology.

\textcolor{black}{Future research directions include extending the proposed learning-based multicast subgrouping framework to multi-cell and \gls{cf-mmimo} deployments, where additional challenges arise, such as inter-cell interference, distributed \gls{csi} acquisition, and decentralized access point coordination. While the present work assumes accurate long-term covariance information, the framework remains compatible with practical systems in which such statistics can be estimated over extended time scales. Incorporating realistic mobility models and explicitly evaluating robustness under time-varying or partially outdated covariance information would further enhance the practical applicability of the approach. Finally, exploring the joint optimization of multicast subgrouping, pilot assignment, and precoding strategies represents a promising direction to unlock additional system-level gains in highly dynamic wireless environments.}

\textcolor{black}{The evaluation presented in this work is based on synthetic channel covariance matrices generated according to widely used local-scattering models, which enable controlled and reproducible analysis. While this allows isolating the impact of spatial channel statistics on the proposed learning-based framework, we do not claim validation on real massive MIMO measurement campaigns. Extending the proposed approach to experimental datasets that capture additional propagation effects, such as Rician fading, non-stationary mobility, or heterogeneous deployments, constitutes an important and relevant direction for future work.}

\appendices

\section{Proof of Lemma 1}
The \gls{mmse} estimate of the subgroup channel $\boldsymbol{h}^g$ can be obtained as
\begin{align}
           \hat{\boldsymbol{h}}^g &= \mathbb{E}\left\{\boldsymbol{h}^g \big | \boldsymbol{y}^g\right\} \\
                                      &=\mathbb{E}\left\{\boldsymbol{h}^g(\boldsymbol{y}^g)\herm\right\}\Big(\mathbb{E}\left\{\boldsymbol{y}^g(\boldsymbol{y}^g)\herm\right\}\Big)^{-1} \boldsymbol{y}^g\nonumber,
    \label{eq:comp_chann_est}
\end{align}
where
\begin{equation}
   \mathbb{E}\left\{\boldsymbol{h}^g (\boldsymbol{y}^g)\herm\right\}=K_g \mathbb{E}\left\{\boldsymbol{h}_{k}(\boldsymbol{h}_{k})\herm\right\} = K_g \boldsymbol{R}^g,
\end{equation}
and
\begin{equation}
\begin{split}
   &\mathbb{E}\left\{\boldsymbol{y}^g (\boldsymbol{y}^g)\herm\right\} \\
   &\qquad=\tau_{\mathrm{p}} \sum_{k\in\mathcal{K}_g} P_{k} \mathbb{E}\left\{\boldsymbol{h}_{k}\boldsymbol{h}_{k}\herm\right\} + \mathbb{E}\left\{\boldsymbol{n}_{g}\boldsymbol{n}_{g}\herm\right\} \\
   &\qquad=\tau_{\mathrm{p}} \! \sum_{k\in\mathcal{K}_g} P_{k} \boldsymbol{R}_{k} + \sigma_{\mathrm{u}}^2\boldsymbol{I}_{M}\triangleq\boldsymbol{\Gamma}_{g}.
\end{split}
\end{equation}

The \gls{mmse} estimate is a zero-mean complex Gaussian vector whose correlation matrix can be straightforwardly obtained as
\begin{equation}
    \mathbb{E}\left\{\hat{\boldsymbol{h}}^g\left(\hat{\boldsymbol{h}}^g\right)\herm\right\}=K_g^2 \boldsymbol{R}^g \boldsymbol{\Gamma}_{g}^{-1} \boldsymbol{R}^g.
\end{equation}
The subgroup channel estimation error $\tilde{\boldsymbol{h}}^g = \boldsymbol{h}^g - \hat{\boldsymbol{h}}^g$ is also a zero-mean complex Gaussian vector whose correlation matrix is given by
\begin{equation}
\begin{split}
    \mathbb{E}\left\{\tilde{\boldsymbol{h}}^g\left(\tilde{\boldsymbol{h}}^g\right)\herm\right\}&=\mathbb{E}\left\{\boldsymbol{h}^g\left(\boldsymbol{h}^g\right)\herm\right\} - \mathbb{E}\left\{\hat{\boldsymbol{h}}^g\left(\hat{\boldsymbol{h}}^g\right)\herm\right\} \\
    &= \boldsymbol{R}^g -  K_g^2\boldsymbol{R}^g \ \boldsymbol{\Gamma}_{g}^{-1} \boldsymbol{R}^g.
    \end{split}
\end{equation}

\bibliographystyle{IEEEtran}
\bibliography{main}
\end{document}